\documentclass[runningheads]{llncs}
\usepackage[T1]{fontenc}
\usepackage{lmodern}
\usepackage{microtype}
\usepackage{amsmath,amssymb,booktabs,array}
\usepackage[hidelinks]{hyperref}
\newcommand{\F}{\mathbb F_q}
\newcommand{\G}{\mathbb G}
\newcommand{\GT}{\mathbb G_T}
\newcommand{\E}{\mathsf E}
\newcommand{\Span}{\operatorname{span}}
\newcommand{\rank}{\operatorname{rank}}
\newcommand{\uid}{\mathsf{id}}
\newcommand{\pk}{\mathsf{PK}}
\newcommand{\sk}{\mathsf{SK}}
\newcommand{\msk}{\mathsf{MSK}}

\newcommand{\sample}{\mathrel{\overset{\$}{\leftarrow}}}
\begin{document}
\title{Cryptanalysis and Repair of
	``Decentralized Multi-authority ABE for NC$^1$ from BDH''}

\titlerunning{Cryptanalysis and Repair of Decentralized MA-ABE}

\author{Xingwang Wang}

\authorrunning{X. Wang}

\institute{
	Shanghai Key Laboratory of Trustworthy Computing,\\
	East China Normal University, Shanghai 200062, China\\
	\email{wxw@stu.ecnu.edu.cn}
}

\maketitle
\begin{abstract}
	We present a plaintext-recovery attack on both the DBDH-based
	construction and the CBDH-based variant of Datta, Komargodski,
	and Waters (Journal of Cryptology, 2023).
	The attack succeeds with probability one in the claimed static
	security model using one corrupt authority and two users,
	each unauthorized even when granted the corrupt authority's
	attribute.
	The flaw is an invalid extension of the LSSS matrix-normalization
	lemma originating with Rouselakis and Waters (FC,~2015):
	preserving complete-share distributions does not preserve
	coordinate-wise ciphertext masks, which a corrupt authority
	can expose to remove identity binding from other users' shares.
	For both variants, replacing these masks with independent full
	linear sharings of zero preserves correctness and restores
	ciphertext-distribution invariance.
	We prove static security of both repaired variants
	in the random oracle model under their respective
	BDH assumptions, using the original hardcore-bit
	convention for the CBDH variant.
	Both proofs allow adaptive identity-hash queries.
	The modification retains the original matrix-width and
	row-labeling restrictions, keys, decryption, and ciphertext size,
	requiring only additional encryption randomness and
	finite-field arithmetic.
	We also attack both multi-authority attribute-based inner-product
	functional encryption constructions of Datta and Pal (PKC, 2023)
	and establish correctness and matrix-substitution invariance
	of the same repair.
	We prove static security of the repaired small-universe
	construction under asymmetric DBDH.
	For the large-universe construction, we identify a missing
	source-group encoding in the original reduction and give
	a complete static security proof under an augmented
	$L$-DBDH assumption, whose generic bilinear group hardness
	we establish.
\end{abstract}

\section{Introduction}
Attribute-based encryption (ABE) associates decryption rights with attributes rather than with an explicit list of recipients~\cite{SW05,GPSW06,BSW07}. A ciphertext-policy ABE (CP-ABE) scheme encrypts under an access policy, and a user decrypts if its attributes satisfy that policy. Crucially, users who are individually unauthorized must not gain decryption capability by pooling their keys. Multi-authority ABE (MA-ABE), introduced by Chase~\cite{Chase07}, distributes attribute certification across different authorities. Lewko and Waters~\cite{LW11} developed decentralized MA-ABE, in which authorities generate their secret keys independently. Their user keys are tied to a common user identifier to prevent collusion across identities.

Linear secret-sharing schemes (LSSS) provide an efficient language for expressing access policies. A matrix $M$ shares a secret $z$ through $M\boldsymbol v$, where the first coordinate of $\boldsymbol v$ is $z$ and its remaining coordinates are uniformly random. Authorized sets reconstruct the secret by a linear combination of their shares. Security arguments often change the basis of the sharing space to isolate rows controlled by corrupt authorities~\cite{RW15,DKWLWE21,DKW23}. Applying such a transformation within an encryption proof requires examining how it affects all observable ciphertext components, rather than only the underlying secret shares.

\paragraph{From share invariance to an invalid ciphertext substitution.}
The matrix-normalization lemma originating with Rouselakis and Waters~\cite{RW15} permits a change of basis that preserves the joint distribution of complete LSSS shares while placing corrupt-authority rows in a form useful for simulation. Writing the transformed matrix as $N=MP$, the share-level guarantee follows by absorbing $P$ into the sharing vector, where $P$ is invertible and preserves its prescribed first coordinate. The security argument of Datta, Komargodski, and Waters~\cite{DKW23} extends this guarantee to the entire challenge ciphertext. However, their ciphertext contains separately encoded terms $M_{i,k}x_k$, in addition to complete shares $M_i\boldsymbol v$. Preserving the distribution of $M\boldsymbol v$ does not establish that of the coordinate-wise array $\{M_{i,k}x_k\}_{i,k}$. A change of basis can mix coordinates, whereas each scalar mask $x_k$ remains attached to a single coordinate. Consequently, the asserted equality of ciphertext distributions does not follow from the normalization lemma.

This application error is distinct from the earlier correction of the proof of the normalization lemma of Rouselakis and Waters~\cite{RW15}. Datta, Komargodski, and Waters~\cite{DKWLWE21} corrected a matrix-inversion argument in that proof and supplied a valid normalization transformation. Our analysis accepts the corrected share-level guarantee and shows that it does not justify the subsequent ciphertext-level substitution in their Journal of Cryptology construction.

\paragraph{From the proof gap to a plaintext-recovery attack.}
The main LSSS construction of Datta, Komargodski, and Waters~\cite{DKW23} claims security under the decisional bilinear Diffie--Hellman (DBDH) assumption. Its ciphertext contains components of the form $g^{y_{\rho(i),k}r_i+M_{i,k}x_k}$, together with $g^{r_i}$. Here $\rho(i)$ labels row $i$, $y_{\rho(i),k}$ is an authority secret exponent, and $r_i$ is row randomness. A corrupt authority knows $y_{\rho(i),k}$ and can therefore remove the authority-dependent factor. Whenever $M_{i,k}\ne0$, raising the remaining element to $M_{i,k}^{-1}$ reveals $g^{x_k}$ without computing a discrete logarithm.

This observation yields a complete attack when the corrupt rows expose the coordinate masks needed by a set of colluding users. The adversary can then remove the identity-dependent terms from their decryption shares and combine the resulting shares across identities. We give an explicit three-row conjunction policy for which every queried identity remains unauthorized even when its attributes are combined with those of the corrupt authority. Nevertheless, one corrupt authority and two such users suffice to recover the challenge plaintext with probability one. The attack follows the original static commitment requirements and uses a matrix whose entries are only zero or one.

\paragraph{Our repair.}
For each hash coordinate $k$, we replace $M_{i,k}x_k$ by $M_i\boldsymbol w^{(k)}$, where $\boldsymbol w^{(k)}$ is sampled independently and uniformly subject to $w^{(k)}_1=0$. The same vector $\boldsymbol w^{(k)}$ is used across all rows for that coordinate. Authorized reconstruction cancels every such zero sharing, preserving the original decryption algorithm.

The modification also restores the distributional property required by the proof. Under an admissible change of basis, the sharing vector and all zero-sharing vectors can be transformed jointly while preserving their sampling distributions and independence. This gives equality of the complete ciphertext payloads under a suitable coupling, with the original policy retained in the ciphertext header. We then give a direct static DBDH reduction in the random oracle model. The repair preserves authority keys, user keys, decryption, and ciphertext size, at the cost of additional encryption randomness and finite-field arithmetic.

\paragraph{Extension to attribute-based inner-product functional encryption.}
The same issue arises in the small- and large-universe multi-authority attribute-based inner-product functional encryption (MA-ABIPFE) constructions of Datta and Pal~\cite{DP23}. Both retain coordinate-wise masking components, and their security arguments likewise extend share-level matrix normalization to the challenge ciphertext. We show that their additional functional-encryption components do not prevent the corresponding removal of identity-dependent masks. The attacks distinguish two one-dimensional challenge
vectors using key queries whose attribute sets remain
unauthorized even when combined with the corrupt attributes.

\paragraph{Contributions.}
Our contributions are as follows.
\begin{itemize}
	\item We give a probability-one plaintext-recovery attack
	against both the DBDH- and CBDH-based variants
	of~\cite{DKW23}.
	The attack uses a three-row LSSS policy, one corrupt
	authority, and two users, each unauthorized even when
	combined with the corrupt authority's attribute.
	It respects the original static commitment requirements.
	
	\item We identify the invalid transition from share-level
	normalization to ciphertext-level distributional
	equivalence and give an explicit counterexample
	(Section~\ref{sec:rootcause}).
	Appendix~\ref{app:originalproof} reconstructs the original
	security argument and isolates the failed implication.
	
	\item We replace coordinate-wise masks with independent
	full sharings of zero and prove correctness and joint
	payload invariance for both variants.
	We give a direct static DBDH reduction and extend
	the simulation to CBDH using the original hardcore-bit
	convention
	(Sections~\ref{sec:repair} and~\ref{sec:security}).
	Both reductions handle maliciously chosen corrupt-authority
	public keys, all committed key queries, and adaptive
	identity-hash queries, as in Remark~5.1 of~\cite{DKW23}.
	
	\item We extend the attack and zero-sharing modification
	to both multi-authority constructions of Datta and
	Pal~\cite{DP23}
	(Appendix~\ref{app:pkc}).
	We establish correctness and payload invariance for both
	modifications and prove static security for the
	small-universe variant.
	For the large-universe variant, we identify a separate
	simulation gap and close the reduction under augmented $L$-DBDH,
	whose generic bilinear group hardness we establish.
\end{itemize}

\paragraph{Scope and attribution.}
Our targets are the construction in Section~4 and its
CBDH-based variant in Remark~4.1 of the published
Journal of Cryptology version of~\cite{DKW23}.
The attack succeeds in the static game of Section~3.3.
The attack also applies to the extension in
Remark~5.1 of~\cite{DKW23}, which permits adaptive
random-oracle queries while retaining all other
initial commitments.
The attack matrix satisfies the scheme's stated LSSS
input restrictions.
For implementations using a fixed policy compiler,
applicability depends on the matrices produced by that compiler.
The main-body repair retains a fixed maximum matrix width
and an injective row-labeling map, so each authority labels
at most one row.

\paragraph{Organization.}
Section~\ref{sec:related} reviews related work.
Section~\ref{sec:prelim} introduces the notation and static
security model, and Section~\ref{sec:target} recalls the
target construction.
Section~\ref{sec:cryptanalysis} presents the attack and
identifies the invalid matrix substitution.
Sections~\ref{sec:repair} and~\ref{sec:security} give the
zero-sharing repair and its security proof.
Section~\ref{sec:efficiency} analyzes its costs.
Appendix~\ref{app:pkc} extends the analysis to the two
MA-ABIPFE constructions and proves security of their modifications,
using an augmented assumption for the large-universe variant.
Appendix~\ref{app:originalproof} reconstructs the original
MA-ABE security argument.

\section{Related Work}\label{sec:related}

\paragraph{Attribute-based encryption.}
Sahai and Waters~\cite{SW05} introduced ABE, followed by
key-policy and ciphertext-policy realizations~\cite{GPSW06,BSW07}.
Waters~\cite{Waters11} developed expressive CP-ABE constructions
under concrete bilinear assumptions, including a DBDH-based
variant with parameters depending on the maximum sharing width.
More recent work has pursued different expressiveness and
efficiency goals.
For example, Cini and Wee~\cite{CW23} construct lattice-based ABE
for circuits with key size independent of the circuit depth.

\paragraph{Decentralized multi-authority ABE.}
Chase~\cite{Chase07} introduced MA-ABE, and Lewko and
Waters~\cite{LW11} developed decentralized MA-ABE.
Rouselakis and Waters~\cite{RW15} obtained an efficient statically
secure large-universe construction and introduced simulation
techniques relevant to subsequent schemes.
Datta, Komargodski, and Waters studied decentralized constructions
from learning with errors (LWE)~\cite{DKWLWE21} and standard bilinear Diffie--Hellman (BDH) assumptions~\cite{DKW23}.
Their fully adaptive construction~\cite{DKWAdaptive23} additionally
accommodates adaptive authority corruption and key queries.
Waters, Wee, and Wu~\cite{WWW22} construct MA-ABE without random
oracles from evasive LWE.
These works address different combinations of assumptions,
policy classes, and corruption models.

\paragraph{Registered and multi-authority registered ABE.}
Hohenberger et al.~\cite{HLWW23} introduced registered ABE, in
which users generate their own keys and register their public
keys and attributes with a transparent key curator.
The curator aggregates public keys and provides auxiliary
decryption information without holding a master decryption secret.
Zhu et al.~\cite{ZZGQ23} develop registered ABE via predicate
encodings in prime-order bilinear groups, using dual-system
techniques and arguments for handling malicious public keys.
Garg et al.~\cite{GLWW24} reduce the common reference string
(CRS) size and give a partitioning-based construction with
static security whose CRS size is independent of the
attribute-universe size.
Lu, Waters, and Wu~\cite{LWW25} extend the registered setting
to multiple independent key curators, with a pairing-based
construction for a bounded number of users and an
obfuscation-based construction supporting an unbounded number.
These models change who generates secret material and which
public values may be adversarially chosen.

\paragraph{Coordinate-wise constructions and corruption.}
The DBDH-based CP-ABE construction in Section 6 of the full
version of Waters~\cite{Waters11} also contains coordinate-wise
products of matrix entries and sharing randomness.
Such a representation alone therefore does not identify the
vulnerability studied here.
In that construction, the attribute bases are generated by the
trusted global setup, and the security experiment does not
expose their discrete logarithms through authority corruption.
Its selective-security proof embeds the committed challenge
matrix into the public parameters.
By contrast, the corrupt-authority interface in~\cite{DKW23}
exposes exponents that allow the adversary to remove
authority-dependent ciphertext masks and obtain reusable
coordinate-mask encodings.
Our extraction attack exploits this additional information.
It does not directly transfer to the single-authority
construction merely because the two schemes use related
algebraic components.

\paragraph{Matrix normalization and cryptanalysis.}
The full version of the LWE work of Datta, Komargodski,
and Waters~\cite{DKWLWE21} corrects the original matrix-normalization proof of
Rouselakis and Waters~\cite{RW15}.
We use the resulting share-level guarantee and identify a
separate application error: extending that guarantee to
ciphertext components whose coordinate-wise masking
distribution need not be preserved.
Boyen's lattice ABE~\cite{Boyen13} provides another example of
the distinction between a proof gap and an actual security
failure; Agrawal et al.~\cite{ABNXXY21} give a polynomial-time
attack against that construction.
Our work supplies both an explicit attack on the pairing-based
target and a modification of its masking randomness with a
complete static-security proof.

\paragraph{Attribute-based inner-product functional encryption.}
Datta and Pal~\cite{DP23} construct MA-ABIPFE schemes combining
decentralized access control with inner-product evaluation on
vectors of unbounded length.
Their small- and large-universe constructions rely on DBDH
and its parameterized $L$-DBDH variant, respectively;
$L$ indexes the number of slots supplied by the latter assumption.
Both contain coordinate-wise masking components to which our
analysis extends.
Appendix~\ref{app:pkc} gives admissible attacks, corresponding
zero-sharing modifications, and a complete security proof for
the repaired small-universe variant.
The technical analysis uses the numbering of their full
version~\cite{DP23Full}.
For the large-universe variant, we close a separate auxiliary-component
simulation gap under augmented $L$-DBDH; this is a stronger assumption
than the one stated in the source.

\paragraph{Summary of results.}
Table~\ref{tab:results} summarizes our results for both
DKW variants and the two DP constructions. The attacks concern the original schemes;
correctness, payload invariance, and security refer to our
full zero-sharing modifications. For the large-universe
MA-ABIPFE construction, the security reduction uses augmented
$L$-DBDH as defined in Appendix~\ref{app:pkc-augmented}.

\begin{table}[t]
	\caption{Summary of attacks and repair guarantees.}
	\label{tab:results}
	\centering
	\small
	\setlength{\tabcolsep}{5pt}
	\renewcommand{\arraystretch}{1.15}
	\begin{tabular}{@{}lccc@{}}
		\toprule
		Property
		& DKW~\cite{DKW23}
		& DP-SU~\cite{DP23}
		& DP-LU~\cite{DP23} \\
		\midrule
		Original scheme
		& Broken & Broken & Broken \\
		
		Repair correctness
		& Proved & Proved & Proved \\
		
		Payload invariance
		& Proved & Proved & Proved \\
		
		Static security
		& Proved & Proved & Proved \\
		
		Assumption
		& DBDH / CBDH & Asym.\ DBDH & Aug.\ $L$-DBDH \\
		
		Analysis
		& \S\ref{sec:cryptanalysis}--\ref{sec:security}
		& App.~\ref{app:pkc}
		& App.~\ref{app:pkc} \\
		\bottomrule
	\end{tabular}
	
	\smallskip
	\begin{minipage}{\linewidth}
		\footnotesize
		DKW: the DBDH- and CBDH-based MA-ABE variants of
		Datta, Komargodski, and Waters;
        DP-SU/DP-LU: the small-/large-universe MA-ABIPFE
        constructions of Datta and Pal. Asym.: asymmetric.
		The attacks recover plaintexts for DKW and break
		indistinguishability for both DP variants.
		The CBDH result uses the hardcore-bit convention
		of Remark~4.1 of~\cite{DKW23}.
		All remaining rows concern our modifications.
		Security is proved in the random oracle model under the stated
		static models and interface restrictions.
		Aug.: augmented with $\{g_2^{c/\mu_i}\}_{i\in[L]}$.
		The large-universe result uses this stronger assumption;
		see Appendix~\ref{app:pkc-large-proof}.
	\end{minipage}
\end{table}

\section{Preliminaries and Security Model}\label{sec:prelim}
\subsection{Bilinear Groups and Linear Secret Sharing}
Let $\lambda$ be the security parameter. Let $q$ be prime, let $\G$ and $\GT$ be cyclic groups of order $q$, let $g$ generate $\G$, and let $e:\G\times\G\to\GT$ be a non-degenerate symmetric bilinear map. Write $\F$ for the field with $q$ elements and set $\E=e(g,g)$, a generator of $\GT$. All exponents lie in $\F$. DBDH asks whether, given $(g,g^a,g^b,g^c,T)$ for independent uniform $a,b,c\in\F$, the element $T$ is $\E^{abc}$ or an independent uniform element of $\GT$. We use the difference-of-probabilities convention for DBDH distinguishing advantage.

Fix a polynomially bounded setup width $d=d(\lambda)$. Write $[n]=\{1,\ldots,n\}$ and use $\sample$ for uniform sampling. Bold vectors are columns unless transposed. For any matrix $M$, let $M_i$, $M_{i,k}$, and $M_I$ denote row $i$, entry $(i,k)$, and the submatrix of rows indexed by $I$, respectively. Write $\boldsymbol e_k$ for the $k$th standard basis vector of $\F^d$. Write $\mathrm{GL}_d(\F)$ for the invertible $d\times d$ matrices.

For the main-body constructions, an LSSS policy is
$(M,\rho)$, with a sharing matrix $M\in\F^{\ell\times d}$ and an injective row-labeling map $\rho:[\ell]\to\mathcal U$, where $\mathcal U$ is the authority universe. Each authority controls one attribute. Smaller matrices are padded with zero columns to the fixed setup width $d$. For $I\subseteq[\ell]$, authorization means that $\boldsymbol e_1^\top\in\Span\{M_i:i\in I\}$. Reconstruction coefficients $\omega_i$ satisfy $\sum_{i\in I}\omega_iM_i=\boldsymbol e_1^\top$. An unauthorized set admits a vector $\boldsymbol t$ satisfying $t_1=1$ and $M_I\boldsymbol t=0$. Both statements follow by elementary linear algebra.

For a secret $z\in\F$, define the zero-sharing subspace $Z=\{\boldsymbol w\in\F^d:w_1=0\}$ and the affine sharing space $V_z=\{\boldsymbol v\in\F^d:v_1=z\}$. A sharing of $z$ samples $\boldsymbol v$ uniformly from $V_z$; a zero sharing samples a vector uniformly from $Z$.

\subsection{Static Security}\label{sec:static}
We use the batched static model of~\cite{DKW23,RW15}. After receiving global parameters, an adversary commits to: a corrupt authority set $\mathcal C$ and its public keys; a disjoint honest authority set $\mathcal N$; key requests $(\uid,U_{\uid})$, where $\uid\in\{0,1\}^*$ is a user identifier and $U_{\uid}\subseteq\mathcal N$ lists its requested authorities; all random oracle inputs; two distinct messages $m_0,m_1\in\GT$ of equal encoding length; and a challenge policy labeled by $\mathcal C\cup\mathcal N$. All requests for one identity are merged. No oracle responses are given before this commitment. For each queried identity, the rows labeled by $\mathcal C\cup U_{\uid}$ must be unauthorized. We also require the corrupt rows alone to be unauthorized, including when the key-query set is empty.

The challenger returns honestly generated public keys for $\mathcal N$, all requested user keys, the committed random oracle responses, and an encryption of $m_\beta$ for a uniformly chosen challenge bit $\beta\in\{0,1\}$. The adversary outputs its guess $\beta'\in\{0,1\}$, with advantage $|\Pr[\beta'=\beta]-1/2|$. Corrupt public keys may be chosen arbitrarily within their specified groups. The challenger need not know their discrete logarithms. Group membership and policy syntax are checked. The random oracle is $H:\{0,1\}^*\times\{2,\ldots,d\}\to\G$, with an injective encoding of its input pair.

This is an indistinguishability under chosen-plaintext attack (IND-CPA) experiment; indistinguishability (IND) advantage below uses the same success-probability convention. This commitment timing is essential for the reduction below: an identity's hash values are programmed after its complete key-request set is known. The theorem does not extend to adaptive key requests by the same proof.

For the CBDH-based variant recalled below, the same
experiment is used with distinct challenge messages
$m_0,m_1\in\{0,1\}$.
The commitment timing, corruption permissions, and
key-query admissibility conditions remain unchanged.

\paragraph{Extension with adaptive random-oracle queries.}
We also consider the extension in Remark~5.1
of~\cite{DKW23}.
All non-oracle commitments above are fixed before
any replies, but the adversary need not commit to
its random-oracle inputs.
After these commitments, it may query $H$ adaptively,
including after receiving the challenge ciphertext.
All answers are consistent with the hash values used
in key generation.
Key requests remain static; in particular, no
additional attributes may subsequently be requested
for a committed identity.

\section{The Target Construction}\label{sec:target}
We recall the DBDH construction in Section~4 of Datta,
Komargodski, and Waters~\cite{DKW23} and its CBDH-based
variant in Remark~4.1, together with the decryption
calculation and the DBDH simulation in their Section~5.2. We use common notation for the construction, attack, and repair, and isolate the matrix-substitution step in the claimed proof. Throughout, $d$ denotes their fixed maximum sharing width $s_{\max}$.

\subsection{Algorithms and Correctness}

Use the policy notation of Section~\ref{sec:prelim}, and let $\mathcal{ID}=\{0,1\}^*$ be the user-identifier space.

\paragraph{$\mathsf{GlobalSetup}(1^\lambda,d)
	\to\mathsf{GP}$.}
On input the security parameter $\lambda$ and sharing
width $d$, generate prime-order bilinear groups
$(q,\G,\GT,g,e)$ and set $\E=e(g,g)$.
Return the global public parameters
$\mathsf{GP}=(q,\G,\GT,g,e,\E,d)$.
The global parameters specify the random oracle
$H:\mathcal{ID}\times\{2,\ldots,d\}\to\G$.
The message space is $\GT$.

\paragraph{$\mathsf{AuthSetup}(\mathsf{GP},u)
	\to(\pk_u,\msk_u)$.}
On input the global parameters and an authority identifier
$u\in\mathcal U$, sample
$\alpha_u,y_{u,2},\ldots,y_{u,d}\sample\F$
independently and set
the public-key components $A_u=\E^{\alpha_u}$ and $Y_{u,k}=g^{y_{u,k}}$ for $2\le k\le d$.
Return the public key
$\pk_u=(A_u,Y_{u,2},\ldots,Y_{u,d})
\in\GT\times\G^{d-1}$
and master secret key
$\msk_u=(\alpha_u,y_{u,2},\ldots,y_{u,d})\in\F^d$.
Honest authorities generate their keys independently.

\paragraph{$\mathsf{KeyGen}(\mathsf{GP},u,\uid,\msk_u)
	\to\sk_{\uid,u}$.}
On input an authority identifier $u$, a user identifier
$\uid\in\mathcal{ID}$, and the corresponding master secret
key $\msk_u$, return the user secret key
\begin{equation}\label{eq:originalkey}
	\sk_{\uid,u}
	=g^{\alpha_u}\prod_{k=2}^d H(\uid,k)^{y_{u,k}}
	\in\G.
\end{equation}
The key grants attribute $u$ to identity $\uid$.

\paragraph{$\mathsf{Enc}(\mathsf{GP},m,(M,\rho),
	\{\pk_u\}_{u\in\rho([\ell])})\to\mathsf{CT}$.}
On input a message $m\in\GT$, a policy $(M,\rho)$,
and the public key of every authority labeling a policy
row, sample
$\boldsymbol v=(z,v_2,\ldots,v_d)\sample\F^d$,
$x_2,\ldots,x_d\sample\F$, and
$r_1,\ldots,r_\ell\sample\F$, all independently.
For $i\in[\ell]$ and $2\le k\le d$, compute
\begin{equation}\label{eq:original}
\begin{aligned}
 C_0&=m\E^z,\ C_{2,i}=g^{r_i},\
 C_{1,i}=\E^{M_i\boldsymbol v}A_{\rho(i)}^{r_i},
 C_{3,i,k}=Y_{\rho(i),k}^{r_i}g^{M_{i,k}x_k}.
\end{aligned}
\end{equation}
Here $C_0,C_{1,i}\in\GT$ and
$C_{2,i},C_{3,i,k}\in\G$.
Return the ciphertext
\[
\mathsf{CT}
=\bigl((M,\rho),C_0,
\{(C_{1,i},C_{2,i},(C_{3,i,k})_{k=2}^d)\}_{i=1}^{\ell}
\bigr).
\]
The policy header is $(M,\rho)$; the remaining components form the payload. Each scalar $x_k$ is shared across all rows.

\paragraph{$\mathsf{Dec}(\mathsf{GP},\mathsf{CT},\uid,S,
	\{\sk_{\uid,u}\}_{u\in S})\to m'\text{ or }\bot$.}
On input a ciphertext $\mathsf{CT}$, a user identifier $\uid$,
an attribute set $S\subseteq\mathcal U$, and the corresponding
user secret keys $\{\sk_{\uid,u}\}_{u\in S}$,
parse the policy $(M,\rho)$ from $\mathsf{CT}$ and set
$I=\{i\in[\ell]:\rho(i)\in S\}$.
If $\boldsymbol e_1^\top\notin\Span\{M_i:i\in I\}$,
return the failure symbol $\bot$.
Otherwise find coefficients
$\{\omega_i\in\F\}_{i\in I}$ satisfying
$\sum_{i\in I}\omega_iM_i=\boldsymbol e_1^\top$.
For each $i\in I$, compute the row-decryption value
\begin{equation}\label{eq:rowdefinition}
	R_i(\uid)=
	\frac{C_{1,i}\prod_{k=2}^d e(H(\uid,k),C_{3,i,k})}
	{e(\sk_{\uid,\rho(i)},C_{2,i})}
	\in\GT.
\end{equation}
Return the decrypted message
$m'=C_0/\prod_{i\in I}R_i(\uid)^{\omega_i}\in\GT$.

\paragraph{Correctness.}
For honestly generated keys and ciphertexts, substitution
into~\eqref{eq:rowdefinition} cancels the
authority-dependent factors and gives
\begin{equation}\label{eq:row}
	R_i(\uid)
	=\E^{M_i\boldsymbol v}
	\prod_{k=2}^d e(H(\uid,k),g)^{M_{i,k}x_k}.
\end{equation}
For a common identity $\uid$, reconstruction gives
$\sum_{i\in I}\omega_iM_i\boldsymbol v=z$ and
$\sum_{i\in I}\omega_iM_{i,k}=0$ for every $k\ge2$.
Thus all identity-dependent factors cancel,
$\prod_{i\in I}R_i(\uid)^{\omega_i}=\E^z$,
and the returned message is $m'=m$.
The attack below removes these identity-dependent factors
before combining rows associated with different identities.

\paragraph{CBDH-based variant.}
Remark~4.1 of~\cite{DKW23} describes a CBDH-based variant
with message space $\{0,1\}$.
Let $\mathsf{hc}$ denote the hardcore-bit function used
in that variant, distinguished from the identity hash $H$.
Encryption replaces only the message component by
\begin{equation}\label{eq:original-cbdh}
	C_0=m\oplus\mathsf{hc}(\E^z).
\end{equation}
Global setup, authority keys, user keys, and the ciphertext
components $C_{1,i},C_{2,i},C_{3,i,k}$ remain unchanged.
For an authorized set of rows belonging to a common identity,
decryption reconstructs
\[
K=\prod_{i\in I}R_i(\uid)^{\omega_i}=\E^z
\]
and returns $m=C_0\oplus\mathsf{hc}(K)$.
Thus, both variants use the same underlying masking element
$K$ and differ only in how it protects the message.

\section{Cryptanalysis}\label{sec:cryptanalysis}
\subsection{Extracting Coordinate Masks}
For a corrupt row $i$, the adversary may choose a correctly generated authority key and hence know every $y_{\rho(i),k}$. It computes the residual coordinate-mask encoding
\begin{equation}\label{eq:leak}
 D_{i,k}=C_{3,i,k}/C_{2,i}^{y_{\rho(i),k}}=g^{M_{i,k}x_k}.
\end{equation}
If $M_{i,k}\ne0$, exponentiation by $M_{i,k}^{-1}$ gives the reusable mask encoding $X_k=g^{x_k}\in\G$. 

Suppose corrupt rows expose $X_k$ for every coordinate used in some set of available rows. For each such row $j$ with a key for identity $\uid_j$, the adversary obtains the identity-independent share encoding $S_j\in\GT$ by computing
\begin{equation}\label{eq:unbound}
	S_j=R_j(\uid_j)\Big/
	\prod_{\substack{2\le k\le d\\M_{j,k}\ne0}}
	e(H(\uid_j,k),X_k)^{M_{j,k}}
	=\E^{M_j\boldsymbol v}.
\end{equation}
Zero coefficients need no extracted mask. For a corrupt row, dividing $C_{1,i}$ by $e(C_{2,i},g)^{\alpha_{\rho(i)}}$ also gives $S_i$. These shares have no remaining identity binding.

\subsection{A Complete Three-Authority Attack}
Let $A,B,C\in\mathcal U$ be three distinct authority identifiers. Consider the width-three matrix
\begin{equation}\label{eq:attackmatrix}
 M=\begin{pmatrix}1&1&1\\0&1&0\\0&0&1\end{pmatrix},\qquad
 \rho(1)=A,\quad\rho(2)=B,\quad\rho(3)=C.
\end{equation}
This matrix realizes $A\wedge B\wedge C$. All three rows reconstruct $\boldsymbol e_1^\top$ with coefficients $(1,-1,-1)$; no proper subset does. In particular, rows $\{1,2\}$ cannot cancel the third coordinate while retaining first coordinate one, and rows $\{1,3\}$ cannot cancel the second.

\paragraph{Step 0: choose the challenge and the corrupt authority.}
After receiving the global parameters, choose the policy in~\eqref{eq:attackmatrix}, corrupt authority $A$, and run authority setup locally to retain $\alpha_A,y_{A,2},y_{A,3}$. Submit the resulting public key and request honest public keys for $B$ and $C$. No malformed key or dishonest parameter generation is needed.

\paragraph{Step 1: commit to legal key and hash requests.}
Choose distinct identities $\uid_B,\uid_C$. Request only $\sk_{\uid_B,B}$ and $\sk_{\uid_C,C}$, together with the hash values $H(\uid_B,k)$ and $H(\uid_C,k)$ for $k=2,3$. For the DBDH-based construction, submit any distinct
messages $m_0,m_1\in\GT$.
For the CBDH-based variant, submit $m_0=0$ and $m_1=1$.
 All these choices are committed before receiving the replies. The effective sets for the two identities are $\{A,B\}$ and $\{A,C\}$, respectively, because the corrupt authority can issue keys freely. Neither set satisfies the challenge policy.

\paragraph{Step 2: extract the common coordinate masks.}
Receive the requested keys, hash outputs, and challenge ciphertext. Compute $X_2=C_{3,1,2}/C_{2,1}^{y_{A,2}}=g^{x_2}$ and $X_3=C_{3,1,3}/C_{2,1}^{y_{A,3}}=g^{x_3}$. Both coefficients in the corrupt row are one, so not even a field inversion is required in this example.

\paragraph{Step 3: remove identity binding from the honest rows.}
Evaluate $R_2(\uid_B)$ and $R_3(\uid_C)$ using the respective queried keys. Compute
\begin{equation}\label{eq:attackrecover}
\begin{aligned}
 S_1&=C_{1,1}/e(C_{2,1},g)^{\alpha_A}=\E^{z+v_2+v_3},\\
 S_2&=R_2(\uid_B)/e(H(\uid_B,2),X_2)=\E^{v_2},\\
 S_3&=R_3(\uid_C)/e(H(\uid_C,3),X_3)=\E^{v_3}.
\end{aligned}
\end{equation}
\paragraph{Step 4: recover the masking element and distinguish.}
Compute
\begin{equation}\label{eq:attack-mask}
	K=\frac{S_1}{S_2S_3}=\E^z.
\end{equation}
For the DBDH-based construction, recover
$m^*=C_0/K$.
For the CBDH-based variant, recover
$m^*=C_0\oplus\mathsf{hc}(K)$.
In either case, $m^*=m_\beta$.
Output zero if $m^*=m_0$, and one otherwise.
The distinct challenge messages make this decision
unambiguous.

\subsection{Admissibility, Success Probability, and Scope}

\begin{theorem}\label{thm:attack}
	Both the DBDH-based construction and its CBDH-based
	variant, when accepting arbitrary bounded-width LSSS
	matrices with injective labels, are not statically secure.
	For every width bound $d\ge3$, each variant admits a
	polynomial-time adversary with plaintext-recovery
	probability one and IND advantage $1/2$.
\end{theorem}

\begin{proof}
	The attack above satisfies the commitment and
	admissibility requirements of both experiments.
	Equations~\eqref{eq:leak}--\eqref{eq:attackrecover}
	involve only the unchanged ciphertext components,
	authority secrets, user keys, and identity hashes.
	They therefore hold for both variants and for every
	choice of encryption coins, including zero exponents.
	Equation~\eqref{eq:attack-mask} recovers the complete
	masking element $K=\E^z$, from which Step~4 recovers
	the challenge plaintext in either variant.
	
	For $d>3$, append zero columns to
	\eqref{eq:attackmatrix} and commit to the additional
	identity-hash inputs for all $2\le k\le d$.
	The additional authority-dependent terms cancel in
	\eqref{eq:row}, while the zero matrix entries contribute
	no coordinate masks.
	Thus, plaintext recovery remains exact and takes
	polynomial time.
	Distinct challenge messages yield guessing probability
	one and IND advantage $1/2$.
	\qed
\end{proof}

\paragraph{Adaptive random-oracle queries.}
The attack also succeeds in the extension described
in Remark~5.1 of~\cite{DKW23}.
That extension permits adaptive random-oracle queries
but keeps the corruption choices, key requests,
challenge policy, and challenge messages committed
at the beginning of the game.
Our adversary already fixes all required hash inputs
in advance and therefore remains admissible without
using the additional adaptivity.
Its plaintext-recovery probability remains one.

\paragraph{Running time.}
At width three, the attack uses two key requests, four identity-hash requests, two source-group exponentiations to extract the masks, and a constant number of pairings and target-group operations to obtain and combine the three unbound shares. Each operation has polynomial cost in the security parameter. Padding to a larger width only adds the corresponding hash terms to the public row calculation, so the running time remains polynomial.  The CBDH-based variant additionally requires one evaluation of $\mathsf{hc}$ and one bitwise XOR.

\begin{proposition}[A sufficient condition for mask extraction]
	\label{prop:mask-coverage}
	Let $(M,\rho)$ be a challenge policy and let
	$J=\rho^{-1}(\mathcal C)$ be its corrupt row indices. Define the set of coordinates exposed by these rows
	\[
	K_{\mathcal C}
	=\{k\in\{2,\ldots,d\}:
	M_{i,k}\ne0\text{ for some }i\in J\}.
	\]
	Suppose the adversary commits to key sets $U_{\uid}$ such
	that $J$ and every $\rho^{-1}(\mathcal C\cup U_{\uid})$
	are unauthorized.
	Assume there is a row set $I$ and coefficients $\omega_i$
	satisfying
	$\sum_{i\in I}\omega_iM_i=\boldsymbol e_1^\top$.
	For every $j\in I\setminus J$, suppose there is an identity
	$\uid_j$ with $\rho(j)\in U_{\uid_j}$ and
	\[
	\{k\in\{2,\ldots,d\}:M_{j,k}\ne0\}
	\subseteq K_{\mathcal C}.
	\]
	Then, for either variant, an adversary choosing honestly
	generated corrupt-authority keys can recover the challenge
	plaintext with probability one in polynomial time.
\end{proposition}

\begin{proof}
	For each $k\in K_{\mathcal C}$, choose a corrupt row $i(k)$
	with $M_{i(k),k}\ne0$. Using its known authority exponent,
	compute
	\[
	X_k=
	\left(
	C_{3,i(k),k}/C_{2,i(k)}^{y_{\rho(i(k)),k}}
	\right)^{M_{i(k),k}^{-1}}
	=g^{x_k}.
	\]
	For each $j\in I\setminus J$, use the queried key for
	$\uid_j$ to obtain
	\[
	S_j=
	R_j(\uid_j)\Big/
	\prod_{k\in K_{\mathcal C}}
	e(H(\uid_j,k),X_k)^{M_{j,k}}
	=\E^{M_j\boldsymbol v}.
	\]
	The support condition ensures that every remaining
	identity-dependent mask is removed.
	For $i\in I\cap J$, compute
	$S_i=C_{1,i}/e(C_{2,i},g)^{\alpha_{\rho(i)}}
	=\E^{M_i\boldsymbol v}$.
	Consequently, the adversary recovers
	$
	K=\prod_{i\in I}S_i^{\omega_i}=\E^{v_1}.
	$
	It obtains the plaintext as $C_0/K$ in the DBDH-based
	construction and as $C_0\oplus\mathsf{hc}(K)$ in the
	CBDH-based variant.
	All keys and required hash inputs can be committed before
	any responses, and the stated unauthorizedness conditions
	make these requests admissible.
	The calculation holds for every choice of encryption coins.
	\qed
\end{proof}

\paragraph{Generality and limits.}
Proposition~\ref{prop:mask-coverage} gives a sufficient
condition, not a necessary characterization of insecurity.
It already uses honestly generated corrupt-authority keys,
so the attack also applies when the model permits maliciously
formed keys.
Its direct target is the matrix-level interface.
Applying it to a restricted policy compiler additionally
requires establishing that the compiler can output a matrix
satisfying the stated conditions.

\subsection{Why the Matrix Substitution Fails}\label{sec:rootcause}
Let
\begin{equation}\label{eq:counterbasis}
 P=\begin{pmatrix}1&0&0\\-1&1&0\\0&-1&1\end{pmatrix},\qquad
 N=MP=\begin{pmatrix}0&0&1\\-1&1&0\\0&-1&1\end{pmatrix}.
\end{equation}
The matrix $P$ is invertible and satisfies $\boldsymbol e_1^\top P=\boldsymbol e_1^\top$. Thus $M\boldsymbol v$ and $N\widetilde{\boldsymbol v}$ have the same distribution when their sharing vectors are uniform with the same first coordinate. Nevertheless, with fixed authority keys the original residual $D_{1,2}$ equals $g^{x_2}$ under $M$, whereas it is always the identity under $N$. Testing whether $D_{1,2}=1$ distinguishes the two payload distributions with probability gap $1-1/q$.

A complete share is preserved by the coupling
$\boldsymbol v_M=P\boldsymbol v_N$.
A coordinate mask instead restricts its randomness to
$\Span\{\boldsymbol e_k\}$, which a general
secret-preserving basis change need not preserve.
Thus share-level invariance does not establish
ciphertext-level invariance. 
We address this failure by replacing coordinate-wise masks
with full zero sharings, whose joint payload invariance is
established in Lemma~\ref{lem:invariant}.

\section{The Full Zero-Sharing Repair}\label{sec:repair}
\subsection{Construction and Correctness}
We first repair the DBDH-based construction.
Global setup, authority setup, and key generation are
exactly as in Section~\ref{sec:target}. To encrypt $m$ under $(M,\rho)$, sample $\boldsymbol v\sample\F^d$, $r_i\sample\F$ for $i\in[\ell]$, and independently $\boldsymbol w^{(k)}\sample Z$ for every $k\in\{2,\ldots,d\}$. Output the same policy and component layout, with
\begin{equation}\label{eq:repair}
\begin{aligned}
 C_0&=m\E^{v_1},C_{1,i}=\E^{M_i\boldsymbol v}A_{\rho(i)}^{r_i},C_{2,i}=g^{r_i},
 C_{3,i,k}=Y_{\rho(i),k}^{r_i}g^{M_i\boldsymbol w^{(k)}}.
\end{aligned}
\end{equation}
The index $k$ selects a hash/key component; it no longer selects a single coordinate of the sharing matrix. A ciphertext row still contains one source-group element per such index. Decryption uses exactly the same formula as the target scheme.

\begin{theorem}[Correctness]\label{thm:correct}
The repaired construction decrypts correctly for every authorized set of rows whose keys belong to one identity.
\end{theorem}
\begin{proof}
Substituting~\eqref{eq:repair} and the unchanged key into the row expression cancels all $y_{\rho(i),k}r_i$ terms and gives $R_i(\uid)=\E^{M_i\boldsymbol v}\prod_{k=2}^d e(H(\uid,k),g)^{M_i\boldsymbol w^{(k)}}$. Therefore
\[
 \prod_{i\in I}R_i(\uid)^{\omega_i}
 =\E^{(\sum_i\omega_iM_i)\boldsymbol v}
 \prod_{k=2}^d e(H(\uid,k),g)^{(\sum_i\omega_iM_i)\boldsymbol w^{(k)}}
 =\E^{v_1}.
\]
The last equality uses $w^{(k)}_1=0$ for each $k$. Dividing $C_0$ by this product recovers $m$.\qed
\end{proof}
\paragraph{Application to the CBDH-based variant.}
The same replacement of $C_{3,i,k}$ applies to the
CBDH-based variant, while retaining
$C_0=m\oplus\mathsf{hc}(\E^{v_1})$.
The correctness calculation above still reconstructs
$K=\E^{v_1}$, after which decryption returns
$m=C_0\oplus\mathsf{hc}(K)$.
The coupling in Lemma~\ref{lem:invariant} also preserves
this message component, since it preserves $v_1$.
Thus, the modification preserves correctness and
restores joint payload invariance for both variants.
Section~\ref{sec:security} establishes static security
of both repaired variants.

\subsection{Joint Payload Invariance}

\begin{lemma}[Joint payload invariance]\label{lem:invariant}
Let $N=MP$, with $P\in\mathrm{GL}_d(\F)$ satisfying $\boldsymbol e_1^\top P=\boldsymbol e_1^\top$. For fixed public keys, fixed $m,z$, and fixed row labels, the repaired ciphertext payloads generated using $M$ and $N$ have identical distributions conditional on $v_1=z$. The public policy header is kept equal to $(M,\rho)$ in this comparison.
\end{lemma}
\begin{proof}
Couple the randomness for $M$ and $N$ by $\boldsymbol v_M=P\boldsymbol v_N$, $\boldsymbol w_M^{(k)}=P\boldsymbol w_N^{(k)}$, and identical $r_i$. The map $P$ is a bijection on $V_z$ and on $Z$. It preserves independent uniform sampling on the product $V_z\times Z^{d-1}$. Moreover, $M\boldsymbol v_M=N\boldsymbol v_N$ and $M\boldsymbol w_M^{(k)}=N\boldsymbol w_N^{(k)}$ simultaneously for every $k$. Every payload component is therefore equal under this coupling, including $C_0$. Keys and hash values need no transformation. The conclusion holds for arbitrary fixed valid public-key group elements, without knowing their logarithms.\qed
\end{proof}

The simulator uses this identity to compute the payload
through $N$ while retaining the originally committed
policy header $(M,\rho)$.

\subsection{Corrupt-Row Normalization}
\begin{lemma}[Corrupt-row normalization]\label{lem:normalize}
Let $J=\rho^{-1}(\mathcal C)$ be unauthorized, let $c_0=\rank(M_J)$, and put $h=d-c_0$. One can efficiently find an invertible $P$ with $\boldsymbol e_1^\top P=\boldsymbol e_1^\top$ such that $N=MP$ satisfies $N_{i,k}=0$ for all $i\in J$ and $k\le h$. Its corrupt-row span is exactly $\{0\}^h\times\F^{c_0}$, and $M,N$ have the same authorized row subsets.
\end{lemma}
\begin{proof}
Let $K=\ker M_J$, of dimension $h$. Unauthorizedness gives $\boldsymbol p_1\in K$ with first coordinate one. Choose $\boldsymbol p_2,\ldots,\boldsymbol p_h$ as a basis of $K\cap Z$, and complete these to a basis of $Z$ with $\boldsymbol p_{h+1},\ldots,\boldsymbol p_d$. Taking these vectors as the columns of $P$ gives an invertible matrix with the required first row. Its first $h$ columns lie in $K$, so the corrupt rows of $N$ vanish on these coordinates. Their rank remains $c_0$, hence they span the entire trailing subspace. Finally, $\sum_i\omega_iM_i=\boldsymbol e_1^\top$ if and only if $\sum_i\omega_iN_i=\boldsymbol e_1^\top$, by multiplication by $P$ or $P^{-1}$.\qed
\end{proof}

The normalization lemma is inherited from~\cite{RW15,DKWLWE21,DKW23}; its basis-extension proof is included for completeness. Lemma~\ref{lem:invariant} supplies the additional ciphertext-level guarantee.

\section{Static Security of the Repair}\label{sec:security}
The repair changes only the masking terms in $C_{3,i,k}$,
preserving the authority keys, user keys, and decryption
algorithm.
Lemma~\ref{lem:invariant} supplies the ciphertext-distribution
guarantee missing from the original proof strategy.
The following reduction establishes static security of
the repaired DBDH-based construction by combining this
guarantee with explicit simulations of public keys,
identity hashes, and user keys.

\begin{theorem}
	Fix a polynomially bounded width parameter $d=d(\lambda)$.
	Under the DBDH assumption, the repaired construction is
	statically IND-CPA secure in the random oracle model defined
	in Section~\ref{sec:static} for all LSSS challenge policies $(M,\rho)$
	having at most $d$ columns and an injective row-labeling
	function $\rho$, even when corrupt authorities choose
	arbitrary valid public keys.
\end{theorem}
\begin{proof}
Let $\mathcal A$ be a static adversary with advantage
$\epsilon$ in the game of Section~\ref{sec:static}. Let $(g,g^a,g^b,g^c,T)$ be the DBDH input. The reduction uses its group as the global parameters. After receiving the static commitments, it obtains $N=MP$ from Lemma~\ref{lem:normalize}. All payload calculations below use $N$, but the ciphertext header remains $(M,\rho)$. Write $J=\rho^{-1}(\mathcal C)$ and $h=d-\rank(M_J)$.

\paragraph{Honest public keys.}
For an honest authority $u=\rho(i)$ appearing in the challenge, choose independent uniform $\alpha'_u,y'_{u,2},\ldots,y'_{u,d}\in\F$ and implicitly set
\begin{equation}\label{eq:embed}
 \alpha_u=\alpha'_u+abN_{i,1},\qquad
 y_{u,k}=\begin{cases}y'_{u,k}+aN_{i,k},&2\le k\le h,\\y'_{u,k},&h<k\le d.\end{cases}
\end{equation}
Publish $A_u=\E^{\alpha'_u}e(g^a,g^b)^{N_{i,1}}$ and $Y_{u,k}=g^{y'_{u,k}}(g^a)^{N_{i,k}}$ for $k\le h$, and $Y_{u,k}=g^{y'_{u,k}}$ otherwise. Injectivity of $\rho$ makes $i$ well-defined. Honest authorities absent from the challenge are generated normally, retaining their exponents. Corrupt public keys are used as supplied. Conditional on $a,b$, all honest exponent tuples are independent uniform translations of fresh randomness.

\paragraph{Random oracle.}
For each identity with committed key set $U_{\uid}$, put
$I_{\uid}=\rho^{-1}(\mathcal C\cup U_{\uid})$ and choose
$\boldsymbol t_{\uid}\in\F^d$ satisfying
$
N_{I_{\uid}}\boldsymbol t_{\uid}=0,
\ t_{\uid,1}=1.
$
Such a vector exists because $I_{\uid}$ is unauthorized
for $M$ and hence also for $N$ by
Lemma~\ref{lem:normalize}.
Since $J\subseteq I_{\uid}$ and the rows of $N_J$ span
$\{0\}^h\times\F^{d-h}$, orthogonality forces
$t_{\uid,k}=0$ for every $k>h$.
For an identity without key requests, take
$\boldsymbol t_{\uid}=\boldsymbol e_1$, which is orthogonal
to all corrupt rows of $N$.

For each identity and coordinate choose independent uniform $\eta_{\uid,k}\in\F$ and return
\begin{equation}\label{eq:hash}
 H(\uid,k)=\begin{cases}(g^b)^{t_{\uid,k}}g^{\eta_{\uid,k}},&2\le k\le h,\\g^{\eta_{\uid,k}},&h<k\le d.\end{cases}
\end{equation}
Maintain one table entry per input and use the same entry inside key generation. This also defines any hashes needed by requested keys but not explicitly requested as oracle outputs. Conditional on $b$ and the committed request sets, these are independent uniform elements of $\G$. Programming all coordinates of each identity together creates no distributional dependence because their $\eta$ values are independent.

\paragraph{Secret keys.}
For a queried honest authority absent from the challenge, compute the key normally using its known exponents and the programmed hash values. For $u=\rho(i)$ in the challenge and $u\in U_{\uid}$, output
\begin{equation}\label{eq:simkey}
 g^{\alpha'_u}
 \prod_{k=2}^{h}\left((g^b)^{t_{\uid,k}y'_{u,k}}
 (g^a)^{\eta_{\uid,k}N_{i,k}}g^{\eta_{\uid,k}y'_{u,k}}\right)
 \prod_{k=h+1}^{d}g^{\eta_{\uid,k}y'_{u,k}}.
\end{equation}
To verify this expression, expand the exponent of the desired key as
\begin{equation}\label{eq:keycancel}
\begin{aligned}
 \alpha_u+\sum_{k=2}^d y_{u,k}\log_g H(\uid,k)
 =&\alpha'_u+ab\left(N_{i,1}+\sum_{k=2}^hN_{i,k}t_{\uid,k}\right)\\
 &+\sum_{k=2}^h\left(bt_{\uid,k}y'_{u,k}
       +a\eta_{\uid,k}N_{i,k}+\eta_{\uid,k}y'_{u,k}\right)
       \\
 &+\sum_{k=h+1}^d\eta_{\uid,k}y'_{u,k}.
\end{aligned}
\end{equation}
The coefficient of $ab$ is $N_i\boldsymbol t_{\uid}=0$, since $t_{\uid,1}=1$ and the trailing coordinates vanish. All remaining terms are computable as in~\eqref{eq:simkey}. The logarithms in~\eqref{eq:keycancel} are only notation for checking the identity; the reduction never evaluates them.

\paragraph{An initial challenge payload.}
Choose the challenge bit $\beta$ uniformly. Implicitly set $\boldsymbol v^0=-abc\boldsymbol e_1$ and
\begin{equation}\label{eq:basew}
 \boldsymbol w^{(k),0}=\begin{cases}-ac\boldsymbol e_k,&2\le k\le h,\\\boldsymbol0,&h<k\le d.\end{cases}
\end{equation}
All these masking vectors lie in $Z$. For a corrupt row $i\in J$, choose $r_i^0\sample\F$ and compute
\[
 C_{1,i}^0=A_{\rho(i)}^{r_i^0},\qquad
 C_{2,i}^0=g^{r_i^0},\qquad
 C_{3,i,k}^0=Y_{\rho(i),k}^{r_i^0}.
\]
This is the correct payload for the implicit vectors because $N_{i,1}=0$ and $N_{i,k}=0$ for $k\le h$. No corrupt exponent is needed.

For every honest challenge row $i$, implicitly set $r_i^0=c$ and output $C_{1,i}^0=e(g^c,g)^{\alpha'_{\rho(i)}}$, $C_{2,i}^0=g^c$, and $C_{3,i,k}^0=(g^c)^{y'_{\rho(i),k}}$. Indeed,
\[
 N_i\boldsymbol v^0+\alpha_{\rho(i)}c=\alpha'_{\rho(i)}c,\qquad
 N_i\boldsymbol w^{(k),0}+y_{\rho(i),k}c=y'_{\rho(i),k}c.
\]
For $k\le h$ the second equality cancels $-acN_{i,k}$ against $acN_{i,k}$; for $k>h$ neither term is embedded. Finally set $C_0^0=m_\beta T^{-1}$. The inverse is required by the chosen sign $v_1^0=-abc$.

\paragraph{Full re-randomization.}
Choose all offsets independently: $\delta\boldsymbol v\sample\F^d$, $\delta r_i\sample\F$, and $\delta\boldsymbol w^{(k)}\sample Z$. Compute
\begin{equation}\label{eq:rerand}
\begin{aligned}
 C_0&=C_0^0\E^{\delta v_1},\ C_{2,i}=C_{2,i}^0g^{\delta r_i},\\
 C_{1,i}&=C_{1,i}^0\E^{N_i\delta\boldsymbol v}A_{\rho(i)}^{\delta r_i},
 C_{3,i,k}=C_{3,i,k}^0Y_{\rho(i),k}^{\delta r_i}g^{N_i\delta\boldsymbol w^{(k)}}.
\end{aligned}
\end{equation}
This operation uses only known scalars and public group elements. It corresponds to adding the offsets to each initial random variable. Conditional on the DBDH exponents, the commitments, the public keys, the hash table, and the returned keys, the new $\boldsymbol v$, all $\boldsymbol w^{(k)}$, and all $r_i$ are jointly independent and uniform in their required spaces. In particular, the initially shared value $r_i^0=c$ among honest rows causes no residual correlation after independent $\delta r_i$ are added.

\paragraph{Real and random cases.}
If $T=\E^{abc}$, the payload in~\eqref{eq:rerand} is exactly a fresh repaired encryption under $N$. Lemma~\ref{lem:invariant} shows that this is also exactly the payload distribution of an encryption under the output header $M$. The public keys, random oracle, and secret keys have already been shown to have their joint real distribution.

If $T$ is uniform in $\GT$, every public-key, hash, key, and challenge component other than $C_0$ is independent of $T$ and $\beta$. The value $C_0=m_\beta T^{-1}\E^{\delta v_1}$ is a uniform one-time pad even conditional on all of those components. The adversary guesses $\beta$ with probability $1/2$. The reduction outputs one precisely when $\beta'=\beta$. The absolute difference between its output probabilities in the two DBDH cases is therefore $\epsilon$. There are no aborts or guessed identities.\qed
\end{proof}

\paragraph{Extension to adaptive random-oracle queries.}
The same reduction establishes security in the extension
of Section~\ref{sec:static}, with the same advantage.
All key requests are committed before any oracle replies,
so the simulator can fix $\boldsymbol t_{\uid}$ for every
identity with a key request before answering hashes.
For an identity without a committed key request, it uses
$\boldsymbol t_{\uid}=\boldsymbol e_1$ as above.

The simulator implements \eqref{eq:hash} by lazy sampling.
On the first use of an input $(\uid,k)$, either during
key generation or in an adversarial oracle query, it
samples $\eta_{\uid,k}$, stores the resulting hash value,
and returns the same value on every subsequent use.
Each fresh answer is uniform conditional on the preceding
transcript, since its exponent contains an independent
uniform additive mask.
The real-case simulation therefore remains exact.
In the random case, the transcript through the challenge
is independent of $\beta$.
Conditional on any fixed preceding transcript, each
subsequent oracle reply has the same distribution for
both values of $\beta$.
Induction over these queries shows that the complete
transcript remains independent of $\beta$.
Thus the distinguishing advantage remains $\epsilon$.

\paragraph{Security of the CBDH-based variant.}
We use the CBDH hardcore-bit convention of
Remark~4.1 of~\cite{DKW23}.
Under this convention, the repaired bit-message variant
is statically IND-CPA secure under CBDH, with the same
policy restrictions and adaptive identity-hash queries
as above.

Given a CBDH instance $(g,g^a,g^b,g^c)$, run the preceding
simulation with the formal substitution $a\mapsto-a$,
implemented using $(g^a)^{-1}$.
All public-key, hash, key, and nonmessage ciphertext
components remain computable without the CBDH target.
The initial sharing vector is now
$\boldsymbol v^0=abc\boldsymbol e_1$.
In the re-randomization step, sample
$\delta\boldsymbol v\sample Z$ rather than $\F^d$,
and sample all other offsets as before.
Consequently, the final secret is $z=abc$.

Conditional on $a,b\ne0$, the uniform exponent $c$
makes $z$ uniform in $\F$.
The independent offsets give the public keys and hashes
their real distributions independently of $a,b,c$,
and make the remaining sharing coordinates, zero-sharing
vectors, and row coins jointly uniform conditional on $z$.
Lemma~\ref{lem:invariant} therefore transfers the simulated
payload to the original header $M$.
The excluded event has probability at most $2/q$.

For distinct challenge bits $m_0,m_1$, sample
$u\sample\{0,1\}$ and set $C_0=u$.
Implicitly define $\beta$ by
$m_\beta=u\oplus\mathsf{hc}(\E^{abc})$.
This is a uniform challenge bit, and the resulting view
has the real challenge distribution.
If the adversary returns $\beta'$, output
$u\oplus m_{\beta'}$ as the prediction of
$\mathsf{hc}(\E^{abc})$.
This prediction is correct exactly when $\beta'=\beta$.
Identical challenge messages contribute zero distinguishing
advantage and can be handled by a random prediction.
Thus a nonnegligible distinguishing advantage yields a
nonnegligible hardcore-bit prediction advantage, with
an additive statistical loss of at most $2/q$.
The same lazy-sampling argument handles subsequent
identity-hash queries.

\paragraph{Boundary cases.}
If $\rank(M_J)=0$, then $h=d$ and the trailing-coordinate
constraints are vacuous.
If $h=1$, orthogonality forces
$\boldsymbol t_{\uid}=\boldsymbol e_1$.
Every queried honest row in the challenge therefore has
$N_{i,1}=0$, so the $ab$ term still vanishes;
products indexed by $2\le k\le h$ are empty products.
For $d=2$, the single masking vector has the form
$\boldsymbol w^{(2)}=(0,w_2)^\top$, and hence
$M_i\boldsymbol w^{(2)}=M_{i,2}w_2$.
Thus the repaired and original masking distributions
coincide at width two.
This is consistent with Theorem~\ref{thm:attack},
whose explicit attack requires $d\ge3$.

\section{Efficiency}\label{sec:efficiency}
Table~\ref{tab:cost} compares the original and repaired
DBDH-based constructions, counting group elements
excluding the policy description. Here $a\GT+b\G$ counts $a$ target- and $b$ source-group elements. Both schemes use the same width $d$ and $\ell$ rows. For $t$ rows participating in decryption, a direct row-by-row evaluation uses $td$ pairings: $d-1$ numerator pairings and one denominator pairing per row. Both allow the same aggregation of numerator pairings for a common identity. Thus the repair introduces no additional group elements or pairing calls.

\begin{table}[t]
\caption{Exact size and randomness counts; arithmetic costs are naive dense bounds.}\label{tab:cost}
\centering
\small
\begin{tabular}{@{}lll@{}}
\toprule
Quantity & Original & Repaired\\
\midrule
Authority public key & $1\GT+(d-1)\G$ & same\\
Authority secret key & $d$ field elements & same\\
User key per authority & $1\G$ & same\\
Ciphertext & $(\ell+1)\GT+\ell d\G$ & same\\
Encryption random scalars & $\ell+2d-1$ & $\ell+d+(d-1)^2$\\
Mask exponent arithmetic & $O(\ell d)$ & $O(\ell d^2)$\\
Direct decryption pairings & $td$ & $td$\\
\bottomrule
\end{tabular}
\end{table}

The extra randomness is $(d-1)(d-2)$ field elements. The mask computations form the product of the last $d-1$ columns of $M$ with a dense $(d-1)\times(d-1)$ random matrix. One can compute each scalar inner product first and then perform a single exponentiation by $g$; the repair does not require one group exponentiation per term in the inner product. Sparse policies can reduce the field-arithmetic cost. 

Existing authority public keys can be used for encryption
under the repaired scheme, and existing user keys can
decrypt the resulting ciphertexts.

\section{Conclusion}
We give a plaintext-recovery attack against both the
DBDH-based construction and its CBDH-based variant,
and identify the invalid extension of LSSS share
invariance to coordinate-wise ciphertext masks.
Independent full sharings of zero repair this masking distribution while preserving keys, decryption, and ciphertext size. A direct static DBDH reduction and its hardcore-bit
extension establish security of both repaired variants
in the random oracle model. Analogous attacks affect both MA-ABIPFE constructions of Datta and Pal considered in the appendix. The zero-sharing modification extends to both, with static security proofs under asymmetric DBDH for the small-universe variant and augmented $L$-DBDH for the large-universe variant. The latter assumption supplies a source-group encoding missing from the original reduction and admits the generic bilinear group analysis given in the appendix.

\bibliographystyle{splncs04}
\bibliography{references}
\appendix
\section{Extension to the MA-ABIPFE Constructions of Datta and Pal}\label{app:pkc}
This appendix analyzes the small- and large-universe constructions in Sections 5 and 6 of the 20 April 2023 full version of Datta and Pal~\cite{DP23Full}. We give attacks and zero-sharing modifications for both. We prove static security of the small-universe modification under asymmetric DBDH and of the large-universe modification under augmented $L$-DBDH. Appendix~\ref{app:pkc-large} isolates the original simulation gap; Appendices~\ref{app:pkc-augmented} and~\ref{app:pkc-large-proof} supply the augmented assumption and complete reduction. The source's single-authority construction is outside our scope.

\subsection{Notation, Construction, and Game}\label{app:pkc-notation}
A ciphertext encrypts a message vector $\boldsymbol v$ under an access policy and a coordinate list $\mathcal I$. Attribute keys associated with the same identity, function vector $\boldsymbol u$, and matching coordinate list allow an authorized user to recover the inner product $\boldsymbol v\cdot\boldsymbol u$. An authorized query with the challenge coordinate list
is admissible only if its inner products with the two
challenge messages are equal. In this appendix, $\boldsymbol v$ denotes the message vector, while $B$ contains the sharing randomness.
Use prime-order groups $\G_1,\G_2,\GT$ of order $q$, generators $g_1,g_2$, and an asymmetric pairing $e:\G_1\times\G_2\to\GT$; set $g_T=e(g_1,g_2)$. The sharing width is $d=s_{\max}$. Message and function vectors $\boldsymbol v,\boldsymbol u\in\F^n$ have a coordinate list $\mathcal I=(\iota_1,\ldots,\iota_n)$ of distinct increasing positive integers; abbreviate $v_{\iota_k},u_{\iota_k}$ as $v_k,u_k$. Each query uses its own list length $n$. Indices range over $i\in[\ell]$, $2\le j\le d$, and $k\in[n]$ unless stated otherwise. For independent random oracles $H_1,H_2,H_3$ into $\G_2$, write $h_{1,t,k}=H_1(t\|\iota_k\|\mathcal I)$, $h_{2,j,k}=H_2(j\|\iota_k\|\mathcal I)$, and $h_{3,\uid,\boldsymbol u,j,k}=H_3(\uid\|\boldsymbol u\|j\|\iota_k)$, suppressing the current list in the first two abbreviations. Concatenations use unambiguous encodings.

\paragraph{Original small-universe construction.} In the small-universe scheme, authority $t$ publishes $(g_1^{\alpha_t},\{g_1^{y_{t,j}}\}_{j=2}^d)$, with independent uniform secret exponents $\alpha_t,y_{t,j}\in\F$ for $2\le j\le d$. Each authority controls the single attribute identified by $t$. It issues the function key
\begin{equation}\label{eq:fe-key}
 K_{\uid,t,\boldsymbol u}=\prod_{k=1}^n h_{1,t,k}^{\alpha_tu_k}
 \prod_{j=2}^d\prod_{k=1}^n(h_{2,j,k}h_{3,\uid,\boldsymbol u,j,k})^{y_{t,j}u_k}.
\end{equation}
For an injectively labeled policy $(M,\rho)$, encryption
samples $B\sample\F^{d\times n}$ and
$r_i,x_{j,k},f_j\sample\F$, all independently.
Let $\boldsymbol z^\top$ be the first row of $B$.
We write the ciphertext vectors coordinatewise:
$\boldsymbol C_0=(C_{0,k})_{k=1}^n$,
$\boldsymbol C_{1,i}=(C_{1,i,k})_{k=1}^n$, and
$\boldsymbol C_{3,i,j}=(C_{3,i,j,k})_{k=1}^n$
belong to $\GT^n$.
The components $C_{2,i},C_{4,i,j}$ belong to $\G_1$,
and $K_{\uid,t,\boldsymbol u}\in\G_2$.
The payload is
\begin{equation}\label{eq:fe-original}
\begin{aligned}
 C_{0,k}&=g_T^{v_k+z_k},\ C_{2,i}=g_1^{r_i},
 C_{1,i,k}=g_T^{(M_iB)_k}e((g_1^{\alpha_{\rho(i)}})^{r_i},h_{1,\rho(i),k}),\\
 C_{3,i,j,k}&=e(g_1^{M_{i,j}x_{j,k}+y_{\rho(i),j}r_i},h_{2,j,k}),
 C_{4,i,j}=g_1^{M_{i,j}f_j+y_{\rho(i),j}r_i}.
\end{aligned}
\end{equation}
The row-decryption value for a key of identity $\uid$ and function $\boldsymbol u$ is
\begin{equation}\label{eq:fe-row}
 R_i(\uid,\boldsymbol u)=
 \frac{\displaystyle\prod_{k=1}^n C_{1,i,k}^{u_k}
 \prod_{j=2}^d\prod_{k=1}^n
 C_{3,i,j,k}^{u_k}e(C_{4,i,j},h_{3,\uid,\boldsymbol u,j,k})^{u_k}}
 {e(C_{2,i},K_{\uid,\rho(i),\boldsymbol u})}.
\end{equation}
The ciphertext contains the header $((M,\rho),\mathcal I)$ and the displayed payload. Each key is associated with its identity, attribute, function vector, and coordinate list. Decryption requires a common identity and function vector, matching coordinate lists, and an authorized attribute set; otherwise it returns $\bot$. For reconstruction coefficients $\omega_i$ satisfying $\sum_i\omega_iM_i=\boldsymbol e_1^\top$, the product $\prod_iR_i(\uid,\boldsymbol u)^{\omega_i}$ equals $g_T^{\boldsymbol z\cdot\boldsymbol u}$. Dividing $\prod_k C_{0,k}^{u_k}$ by this product yields $g_T^{\boldsymbol v\cdot\boldsymbol u}$. Efficiently recovering the numerical inner product requires an output range with efficiently computable discrete logarithms. Our attack uses only values zero and one.

\paragraph{Original large-universe construction.}
The large-universe variant has an attribute universe $\mathcal X$, an ownership map $\mathsf T:\mathcal X\to\mathcal U$, and row labels $\delta:[\ell]\to\mathcal X$, with $\rho=\mathsf T\circ\delta$; labels need not be injective. Authority setup is unchanged, and each authority shares its exponents across its attributes. An additional independent random oracle $R$ maps tuples $(t,j,\mathcal I)$, with $t\in\mathcal X$ and $j\in[d]$, to $\G_2$.

For attribute $t$, let $K_{\uid,t,\boldsymbol u}$ denote the base expression in~\eqref{eq:fe-key}, using $\alpha_{\mathsf T(t)},y_{\mathsf T(t),j}$ as the authority exponents and retaining $t$ as the $H_1$ input. Each key-generation call independently samples $\tau_j\sample\F$ for $j\in[d]$ and returns
\[
K^{\mathrm{LU}}_{\uid,t,\boldsymbol u}
=K_{\uid,t,\boldsymbol u}\prod_{j=1}^d
R(t\|j\|\mathcal I)^{\tau_j},
\qquad
Z_t^{(j)}=g_1^{\tau_j}.
\]
The dependence on $\mathcal I$ and on the particular key-generation call is suppressed.

Encryption uses the same coins and coordinatewise notation
as the small-universe scheme.
The components $C_{0,k},C_{2,i},C_{3,i,j,k},C_{4,i,j}$
remain as in~\eqref{eq:fe-original}, while
\[
C_{1,i,k}
=g_T^{(M_iB)_k}
e((g_1^{\alpha_{\rho(i)}})^{r_i},h_{1,\delta(i),k}),
\
C_{5,i,j}=R(\delta(i)\|j\|\mathcal I)^{r_i}\in\G_2
\quad(j\in[d]).
\]
The ciphertext header is $((M,\delta),\mathcal I)$.
For row decryption, use~\eqref{eq:fe-row} with denominator
$e(C_{2,i},K^{\mathrm{LU}}_{\uid,\delta(i),\boldsymbol u})$
and multiply its numerator by
$\prod_{j=1}^d e(Z_{\delta(i)}^{(j)},C_{5,i,j})$.
The auxiliary factors cancel, so reconstruction and the remaining decryption conditions are unchanged.

Here $H_1$ uses attribute $\delta(i)$, whereas authority exponents use $\rho(i)$. This follows the source's correctness calculation; its encryption listing inconsistently uses $\rho(i)$ as the $H_1$ input. The attack below chooses matching attribute and authority identifiers and does not rely on this discrepancy.

\paragraph{Static MA-ABIPFE admissibility.}
Both variants use the static security framework of
Definition~4.1 and Remark~4.1 of~\cite{DP23Full},
with the small-universe restrictions specified
in Remark~4.2.
After global setup, all corrupt public keys, honest authorities, hash inputs, key tuples $(\uid,S,\boldsymbol u,\mathcal I)$, where $S$ is the requested attribute set, and challenge vectors $(\boldsymbol v_0,\boldsymbol v_1,\mathcal I^\star)$ and policy are committed before replies. Here $\mathcal I^\star$ is the common coordinate list of the two challenge messages $\boldsymbol v_0,\boldsymbol v_1\in\F^{|\mathcal I^\star|}$. The challenger encrypts $\boldsymbol v_\beta$ for a uniform bit $\beta$, and the adversary guesses that bit; advantage is measured as in Section~\ref{sec:static}. The key tuples have distinct identities as stipulated in~\cite{DP23Full}. For each tuple with $\mathcal I=\mathcal I^\star$, either its attributes together with all corrupt attributes are unauthorized, or $(\boldsymbol v_0-\boldsymbol v_1)\cdot\boldsymbol u=0$. As a nontriviality condition, the attributes controlled by corrupt authorities alone must not satisfy the challenge policy. Each requested attribute set $S$ contains only attributes
controlled by honest authorities, i.e.,
$\mathsf T(S)\cap\mathcal C=\varnothing$;
in the small-universe setting this reduces to
$S\cap\mathcal C=\varnothing$.
The adversary generates corrupt keys itself and can retain their exponents. The attacks below meet these restrictions without using malformed keys.

\subsection{An Admissible Attack on Both Variants}\label{app:pkc-attack}
Take $d=3$, $\mathcal I=\mathcal I^\star=(1)$, function vector $\boldsymbol u=(1)$, and challenge vectors $\boldsymbol v_0=(0)$, $\boldsymbol v_1=(1)$. Use the matrix in~\eqref{eq:attackmatrix}, corrupt $A$, and request only the $B$ key for $\uid_B$ and the $C$ key for a different identity $\uid_C$. Request all hashes used below. The sets $\{A,B\}$ and $\{A,C\}$ are unauthorized, so the queries are legal even though the two challenge inner products differ. Every choice is made before receiving any response. In the large-universe variant choose three attributes, each controlled by its correspondingly named authority.

Since there is a single message coordinate, write
$\boldsymbol b=(z,b_2,b_3)^\top$ for $B$. The corrupt first row exposes the mask encodings $F_j\in\G_1$ and $X_j\in\GT$, for $j=2,3$,
\begin{equation}\label{eq:fe-extract}
 F_j=\frac{C_{4,1,j}}{C_{2,1}^{y_{A,j}}}=g_1^{f_j},\qquad
 X_j=\frac{C_{3,1,j,1}}{e(C_{2,1},h_{2,j,1})^{y_{A,j}}}
     =e(g_1,h_{2,j,1})^{x_{j,1}}.
\end{equation}
Expanding~\eqref{eq:fe-row} for either queried key gives
\[
 R_i(\uid,(1))=g_T^{M_i\boldsymbol b}
 \prod_{j=2}^3 X_j^{M_{i,j}}e(F_j,h_{3,\uid,(1),j,1})^{M_{i,j}}.
\]
Consequently the adversary computes
\begin{equation}\label{eq:fe-attack}
\begin{aligned}
 S_1&=C_{1,1,1}/e(C_{2,1},h_{1,A,1})^{\alpha_A}=g_T^{z+b_2+b_3},\\
 S_2&=R_2(\uid_B,(1))/(X_2e(F_2,h_{3,\uid_B,(1),2,1}))=g_T^{b_2},\\
 S_3&=R_3(\uid_C,(1))/(X_3e(F_3,h_{3,\uid_C,(1),3,1}))=g_T^{b_3}.
\end{aligned}
\end{equation}
It recovers $C_{0,1}S_2S_3/S_1=g_T^{v_\beta}$. Testing whether this equals the identity or $g_T$ determines $\beta$ exactly. No general discrete logarithm is computed.

In the large-universe case, use its augmented row formula
with $t=\delta(i)$ and the auxiliary exponents $\tau_j$
of the queried key.
The extra numerator factors equal
$\prod_{j=1}^d e(g_1,R(t\|j\|\mathcal I))^{\tau_jr_i}$
and cancel the extra factors in the key denominator.
The residual expression and every step
of~\eqref{eq:fe-extract}--\eqref{eq:fe-attack} are unchanged.

\begin{proposition}\label{prop:fe-attack}
Both multi-authority constructions in Sections 5 and 6 of~\cite{DP23Full}, with their stated matrix-level interfaces, admit a polynomial-time static adversary with success probability one and distinguishing advantage $1/2$. The attack uses one corrupt authority and two legal key-query identities.
\end{proposition}
\begin{proof}
Admissibility follows from the two unauthorized attribute sets. The displayed identities hold for all encryption coins and recover the binary challenge. Extra width is handled by padding zero columns. The large-universe claim follows from the explicit cancellation above, using one attribute per authority as a permitted special case.\qed
\end{proof}

\subsection{Full Zero-Sharing Modification}\label{app:pkc-repair}
We modify only the masking terms in $C_{3,i,j,k}$
and $C_{4,i,j}$ of the original constructions recalled
in Appendix~\ref{app:pkc-notation}.
Independently sample $\boldsymbol w^{(j,k)}\sample Z$ for every $2\le j\le d$, $1\le k\le n$, and $\boldsymbol f^{(j)}\sample Z$ for every $2\le j\le d$. Set
\begin{equation}\label{eq:fe-repair}
\begin{aligned}
 C_{3,i,j,k}&=e(g_1^{M_i\boldsymbol w^{(j,k)}+y_{\rho(i),j}r_i},h_{2,j,k}),\quad
 C_{4,i,j}=g_1^{M_i\boldsymbol f^{(j)}+y_{\rho(i),j}r_i}.
\end{aligned}
\end{equation}
Each vector is shared across rows for its fixed superscript, not resampled separately for each row. All other algorithms and components, including $C_5$ in the large-universe variant, stay unchanged subject to the explicit attribute-index convention above.

\begin{proposition}[Correctness and payload invariance]\label{prop:fe-invariance}
Both modified constructions are correct for authorized keys of a common identity, function vector, and matching index set. For $N=MP$ with $P$ invertible and $\boldsymbol e_1^\top P=\boldsymbol e_1^\top$, their payload laws are invariant under substitution of $M$ by $N$, for fixed public keys, oracle functions, message vector,
row labels, coordinate list, and first row of $B$, while retaining the original policy header.
\end{proposition}
\begin{proof}
Row decryption in either variant leaves
\[
 g_T^{M_iB\boldsymbol u}
 \prod_{j=2}^d\prod_{k=1}^n
 e(g_1,h_{2,j,k})^{u_kM_i\boldsymbol w^{(j,k)}}
 e(g_1,h_{3,\uid,\boldsymbol u,j,k})^{u_kM_i\boldsymbol f^{(j)}}.
\]
Reconstruction applies $\sum_i\omega_iM_i=\boldsymbol e_1^\top$ to every exponent, killing every zero sharing and retaining $g_T^{\boldsymbol z\cdot\boldsymbol u}$. For invariance, couple $B_M=PB_N$, $\boldsymbol w_M^{(j,k)}=P\boldsymbol w_N^{(j,k)}$, and $\boldsymbol f_M^{(j)}=P\boldsymbol f_N^{(j)}$, while using identical $r_i$. These are bijections of their sampling spaces, preserve all independence and shared-vector relations, and make every payload component equal. The components $C_5$ do not depend on the sharing matrix. Equality under this distribution-preserving coupling establishes the claim.\qed
\end{proof}

The two families of masks use $(d-1)(n+1)$ random field scalars originally and $(d-1)^2(n+1)$ after modification. The other coins are unchanged. There are no additional ciphertext group elements or decryption pairing calls. Naive mask arithmetic increases from $O(\ell(d-1)(n+1))$ to $O(\ell(d-1)^2(n+1))$ field operations. 

\subsection{Static Security of the Small-Universe Modification}\label{app:pkc-proof}
For independent uniform $a,b,c\in\F$, the asymmetric DBDH input is $(g_1^a,g_1^c,g_2^a,g_2^b,T)$ together with the group description, with $T$ either $g_T^{abc}$ or independent uniform. This is the input used in Assumption 3.1 of~\cite{DP23Full}.

\begin{theorem}\label{thm:fe-small}
Under this asymmetric DBDH assumption, the small-universe modification is statically IND-CPA secure in the random oracle model specified above, for polynomially bounded sharing width and injective authority labels. An adversary with advantage $\epsilon$ gives a DBDH distinguisher with difference-of-probabilities advantage $\epsilon$.
\end{theorem}
\begin{proof}
Let $\boldsymbol\Delta=\boldsymbol v_0-\boldsymbol v_1\in\F^n$ be the challenge-message difference vector, where $n=|\mathcal I^\star|$. Identical challenge vectors give zero advantage and can be handled trivially. Assume they differ, and let $J$ be the challenge-row indices controlled by corrupt authorities. Normalize the corrupt rows using Lemma~\ref{lem:normalize}, obtaining $N=MP$ and $h=d-\rank(M_J)$. All payload calculations use $N$, but the output header remains $M$.

\paragraph{Authority keys.}
For honest $t=\rho(i)$ in the policy, let $\nu_{t,1}=N_{i,1}$, and let $\nu_{t,j}=N_{i,j}$ for $2\le j\le h$ and zero otherwise. Choose fresh uniform offsets $\widehat\alpha_t,\widehat y_{t,j}$ and implicitly set $\alpha_t=\widehat\alpha_t+a\nu_{t,1}$, $y_{t,j}=\widehat y_{t,j}+a\nu_{t,j}$. Publish the required $\G_1$ elements using $g_1^a$. Authorities absent from the policy are generated honestly. Corrupt public keys are adopted as supplied; their exponents need not be known.

\paragraph{Oracle programming.}
For every relevant attribute $t$ and
$\iota_k\in\mathcal I^\star$, set
$H_1(t\|\iota_k\|\mathcal I^\star)
=(g_2^b)^{\Delta_k}g_2^{\eta_{1,t,k}}$.
For every relevant coordinate list $\mathcal I$,
$\iota_k\in\mathcal I$, and $2\le j\le d$, set
$H_2(j\|\iota_k\|\mathcal I)
=g_2^{\eta_{2,j,k,\mathcal I}}$.
All offsets are sampled independently and uniformly
from $\F$.
All other $H_1$ inputs, including those with
$\mathcal I\ne\mathcal I^\star$, receive independent
known-exponent random elements. For an identity whose committed query has $\mathcal I=\mathcal I^\star$ and whose attributes together with the corrupt attributes are unauthorized, choose $\boldsymbol d_{\uid}$ with first coordinate one, zero trailing coordinates beyond $h$, and orthogonal to all the corresponding rows of $N$. Set $\gamma_{\uid,j}=d_{\uid,j}$ for $2\le j\le d$. For an authorized query set, put $\gamma_{\uid,j}=0$ for all $j$. In both cases sample independent $\eta_{3,\uid,\boldsymbol u,j,k}\sample\F$ and program, at the query's function vector,
\[
 H_3(\uid\|\boldsymbol u\|j\|\iota_k)
 =(g_2^b)^{\Delta_k\gamma_{\uid,j}}g_2^{\eta_{3,\uid,\boldsymbol u,j,k}}.
\]
For queries with a different index set, and for other $H_3$ inputs, use independent known-exponent elements. Distinct queried identities ensure that these rules never assign conflicting values across key tuples, despite the absence of $\mathcal I$ in $H_3$'s input. Every complete oracle input is recorded and answered consistently. Conditional on $b$, independent uniform offsets make all oracle values independent uniform elements with their real distribution.

\paragraph{Key responses.}
Consider a query with $\mathcal I=\mathcal I^\star$ and an honest attribute $t$ in the policy. Define the known scalars (with $2\le j\le d$ for $B_j$)
\[
 D=\boldsymbol\Delta\cdot\boldsymbol u,\quad
 A=\sum_{k=1}^n u_k\eta_{1,t,k},\quad
 B_j=\sum_{k=1}^n u_k(\eta_{2,j,k,\mathcal I^\star}+\eta_{3,\uid,\boldsymbol u,j,k}).
\]
The exponent of the desired key has its only unavailable quadratic term
\[
 abD\left(\nu_{t,1}+\sum_{j=2}^d \nu_{t,j}\gamma_{\uid,j}\right).
\]
For an unauthorized query, the parenthesis equals $N_i\boldsymbol d_{\uid}=0$. For an authorized query, $D=0$ by MA-ABIPFE admissibility. The key is therefore computable as
\begin{equation}\label{eq:fe-simkey}
\begin{aligned}
 &g_2^{\widehat\alpha_t A+\sum_{j=2}^d\widehat y_{t,j}B_j}
 (g_2^a)^{\nu_{t,1} A+\sum_{j=2}^d \nu_{t,j}B_j}
 \cdot(g_2^b)^{D(\widehat\alpha_t+\sum_{j=2}^d\widehat y_{t,j}\gamma_{\uid,j})}.
\end{aligned}
\end{equation}
For a nonmatching index set, every hash used in the queried key has a known exponent, so $g_2^a$ suffices to compute the key. For an authority absent from the policy, its known secret exponents suffice even when the hashes are programmed. Each response is thus exactly the key determined by the simulated public keys and hash table.

\paragraph{Initial challenge.}
Implicitly choose
\[
 B^0=-abc\boldsymbol e_1\boldsymbol\Delta^\top,\qquad
 \boldsymbol w^{(j,k),0}=\boldsymbol f^{(j),0}=-ac\boldsymbol e_j
 \quad(2\le j\le h),
\]
with zero vectors for $j>h$. For an honest row set $r_i^0=c$. Let $t=\rho(i)$. Its initial components are computable as
\begin{equation}\label{eq:fe-base}
\begin{aligned}
 C_{1,i,k}^0={}&e(g_1^c,g_2^b)^{\widehat\alpha_t\Delta_k}
 e(g_1^c,g_2)^{\widehat\alpha_t\eta_{1,t,k}}
 e(g_1^c,g_2^a)^{N_{i,1}\eta_{1,t,k}},\\
 C_{2,i}^0={}&g_1^c,\quad
 C_{3,i,j,k}^0=e((g_1^c)^{\widehat y_{t,j}},h_{2,j,k}),\quad
 C_{4,i,j}^0=(g_1^c)^{\widehat y_{t,j}}.
\end{aligned}
\end{equation}
The $abcN_{i,1}\Delta_k$ term in $C_1$ cancels against $N_iB^0$; the $acN_{i,j}$ terms in $C_3,C_4$ cancel against the zero-sharing inner products. For a corrupt row choose known uniform $r_i^0$. All initial sharing terms vanish on this row, so compute its components directly from its supplied public key and the public hashes. Put $C_{0,k}^0=g_T^{v_{\beta,k}}T^{-\Delta_k}$ for a uniform bit $\beta$.

\paragraph{Re-randomization.}
Sample $\widetilde B\sample\F^{d\times n}$ and each $\widetilde r_i\sample\F$ independently. Independently sample each vector $\widetilde{\boldsymbol w}^{(j,k)}$ and $\widetilde{\boldsymbol f}^{(j)}$ uniformly from $Z$. Let $\widetilde z_k=\widetilde B_{1,k}$. Write $\widetilde{\boldsymbol z}=(\widetilde z_k)_{k=1}^n$. Multiply the initial components by
\begin{equation}\label{eq:fe-rerand}
\begin{aligned}
 D_{0,k}&=g_T^{\widetilde z_k},\ D_{2,i}=g_1^{\widetilde r_i},
 D_{1,i,k}=g_T^{(N_i\widetilde B)_k}
 e((g_1^{\alpha_{\rho(i)}})^{\widetilde r_i},h_{1,\rho(i),k}),\\
 D_{3,i,j,k}&=e(g_1^{N_i\widetilde{\boldsymbol w}^{(j,k)}}
 (g_1^{y_{\rho(i),j}})^{\widetilde r_i},h_{2,j,k}),\
 D_{4,i,j}=g_1^{N_i\widetilde{\boldsymbol f}^{(j)}}
 (g_1^{y_{\rho(i),j}})^{\widetilde r_i}.
\end{aligned}
\end{equation}
These multipliers use only known scalars and public group elements. The effective coins are jointly uniform in the modified scheme's product sampling space even conditional on the previously simulated view and the challenge exponents. Proposition~\ref{prop:fe-invariance} then converts the payload law for $N$ to the law for the unchanged output header $M$.

\paragraph{Real and random targets.}
When $T=g_T^{abc}$, the complete view is exactly real. When $T=g_T^\theta$ for an independent uniform exponent $\theta\in\F$, the only components that involve $\theta$ or $\beta$ are $C_{0,k}$, whose exponent vector is $\boldsymbol v_\beta+\widetilde{\boldsymbol z}-\theta\boldsymbol\Delta$. Even conditional on all other components and on the simulator's other coins, this is uniform along the same affine line for either bit, since
$
 \boldsymbol v_0-\theta\boldsymbol\Delta
 =\boldsymbol v_1-(\theta-1)\boldsymbol\Delta.
$
Thus the adversary's success probability is $1/2$ in the random case. Outputting one exactly when the adversary guesses $\beta$ yields the stated DBDH advantage.\qed

\end{proof}

\subsection{Large-Universe Modification: The Original Simulation Gap}\label{app:pkc-large}
Propositions~\ref{prop:fe-attack} and~\ref{prop:fe-invariance} apply to the large-universe interface, including multiple attributes per authority and repeated labels. However, they do not alone establish its repaired security under the source's $L$-DBDH assumption. We separate the key-cancellation calculation that carries over from a remaining ciphertext-computability issue.

Here $L$ is the number of indexed slots in the source's $L$-DBDH instance; assign the $\ell$ challenge rows distinct slots, so $L\ge\ell$. The symbols $\mu_i$ denote the hidden field exponents associated with those slots, not scalars known to the simulator. Expressions involving $1/\mu_i$ are read on the event $\mu_i\ne0$ for all used slots. Conditioning the source's uniform $\mu_i\in\F$ on nonzero values changes the distribution by at most $L/q$. The offsets $\widehat\alpha_\theta,\widehat y_{\theta,j},\xi_{t,j}$ below are independently sampled from $\F$. Set $J=\{i:\rho(i)\text{ is corrupt}\}$, $h=d-\rank(M_J)$, and obtain $N$ by corrupt-row normalization. For an attribute $t$, let $\theta=\mathsf T(t)$.

Let $X_\theta=\{i:\rho(i)=\theta\}$ and $X_{\theta,t}^{\mathrm{other}}=X_\theta\setminus\{i:\delta(i)=t\}$. The source reduction embeds
\[
 \alpha_\theta=\widehat\alpha_\theta+a\sum_{i\in X_\theta}\mu_iN_{i,1},\qquad
 y_{\theta,j}=\widehat y_{\theta,j}+a\sum_{i\in X_\theta}\mu_iN_{i,j}.
 \quad(2\le j\le h)
\]
For $1\le j\le h$, it programs
$R(t\|j\|\mathcal I^\star)
=g_2^{\xi_{t,j}}
\prod_{i\in X_{\theta,t}^{\mathrm{other}}}
(g_2^{a\mu_i})^{N_{i,j}}$.
For $j>h$, it sets
$y_{\theta,j}=\widehat y_{\theta,j}$ and
$R(t\|j\|\mathcal I^\star)=g_2^{\xi_{t,j}}$.
 With the challenge-difference programming above, an unauthorized query can use initial $\tau_j^0=-b(\boldsymbol\Delta\cdot\boldsymbol u)d_{\uid,j}$ for $1\le j\le h$ (and zero for $j>h$), where $\boldsymbol d_{\uid}\in\F^d$ has first coordinate one and is orthogonal to the normalized rows labeled by the queried and corrupt attributes. The quadratic terms associated with other attributes cancel, leaving
\[
 ab(\boldsymbol\Delta\cdot\boldsymbol u)
 \sum_{i:\delta(i)=t}\mu_iN_i\boldsymbol d_{\uid}=0.
\]
For authorized admissible queries the factor $\boldsymbol\Delta\cdot\boldsymbol u$ is zero. Independent key re-randomization can restore the required distribution of $\tau_j$. These observations explain why function-key cancellation is compatible with the proposed modification.

For the initial challenge, the source uses $r_i^0=c/\mu_i$.
For $2\le j\le h$, our initial zero-sharing vectors
produce $-acN_{i,j}$ and preserve the intended diagonal
cancellation. However, the unchanged auxiliary component
must also be computed. For $1\le j\le h$, it equals
\begin{equation}\label{eq:large-obligation}
	C_{5,i,j}^0=
	g_2^{\xi_{\delta(i),j}c/\mu_i}
	\prod_{s\in X_{\rho(i),\delta(i)}^{\mathrm{other}}}
	(g_2^{ac\mu_s/\mu_i})^{N_{s,j}}.
\end{equation}
For $j>h$, only the first factor remains.
The off-diagonal product is available in the stated
$L$-DBDH input. The first factor is not directly supplied:
Assumption 3.2 of~\cite{DP23Full} supplies
$g_1^{c/\mu_i}$ but does not list $g_2^{c/\mu_i}$,
nor an efficient source-group map that would produce it.
Treating $\xi_{\delta(i),j}$ as a known scalar therefore
does not by itself compute this factor in $\G_2$.
Pairing can produce target-group encodings, which do not
substitute for the required source-group element.
Full zero-sharing re-randomization cannot repair an
initial component that has not been computed.

The source reduction requires the same unavailable
source-group element.
We complete the simulation under an augmented assumption
that additionally supplies
$\{g_2^{c/\mu_i}\}_{i\in[L]}$.
The next two subsections establish generic hardness of
this assumption and give a static security reduction
for the full zero-sharing modification.

\subsection{An Augmented Assumption and Its Generic Hardness}\label{app:pkc-augmented}
The augmentation changes the computational assumption,
not the modified scheme.

\paragraph{Augmented $L$-DBDH assumption.}
Let $(\G_1,\G_2,\GT,g_1,g_2,e)$ be asymmetric prime-order bilinear groups, with $g_T=e(g_1,g_2)$. Choose $a,b,c\sample\F$ and independently $\mu_1,\ldots,\mu_L\sample\F^*$. Define the following sets of formal exponents:
\begin{equation}\label{eq:aug-exponents}
\begin{aligned}
 \mathcal P={}&\{1,b,c\}\cup\{a\mu_i,c/\mu_i:i\in[L]\}
 \cup\{c\mu_s/\mu_i,ac\mu_s/\mu_i:i,s\in[L],\ i\ne s\},\\
 \mathcal Q^+={}&\{1,a,b\}\cup\{a\mu_i,c/\mu_i:i\in[L]\}\cup\{ac\mu_s/\mu_i:i,s\in[L],\ i\ne s\}.
\end{aligned}
\end{equation}
The public instance is $\mathcal D_L^+=(\mathsf{desc},\{g_1^p:p\in\mathcal P\},\{g_2^q:q\in\mathcal Q^+\})$, where $\mathsf{desc}$ includes the group description and generators. Augmented $L$-DBDH asserts computational indistinguishability of
$
 (\mathcal D_L^+,g_T^{abc})\ \hbox{and}\
 (\mathcal D_L^+,g_T^\vartheta),\ \vartheta\sample\F.
$
Its distinguishing advantage is the absolute difference between the probabilities of outputting one in these experiments. The only new source-group family, relative to Assumption 3.2 of~\cite{DP23Full}, is $\{g_2^{c/\mu_i}\}_{i\in[L]}$. Sampling the denominators in $\F^*$ avoids undefined inputs; conditioning the source's sampling on this event costs at most $L/q$ in statistical distance.

\begin{proposition}[Generic hardness of augmented $L$-DBDH]\label{prop:aug-generic}
In the asymmetric generic bilinear group model with independent source-group encodings and no source-group conversion oracle, a generic algorithm making at most $Q$ oracle operations has augmented $L$-DBDH advantage
\[
 O\!\left(\frac{(L+1)(Q+L^2)^2}{q-1}\right).
\]
In particular, the advantage is negligible when $L,Q$ are polynomially bounded in the security parameter and $q$ is exponentially large.
\end{proposition}
\begin{proof}
Work first in the Laurent polynomial ring
$\F[a,b,c,\mu_1^{\pm1},\ldots,\mu_L^{\pm1}]$.
Formal source-group exponents lie in the linear spans of $\mathcal P$ and $\mathcal Q^+$; target-group exponents, excluding the challenge, lie in
$
 \mathcal V=\Span_{\F}\bigl(\{1\}\cup
 \{p q:p\in\mathcal P,\ q\in\mathcal Q^+\}\bigr).
$
Neither source set contains a monomial involving $b$ other than $b$ itself. Thus a product of two listed source monomials could equal $abc$ only by multiplying $b$ and $ac$. Neither set contains $ac$: its apparent substitutes $ac\mu_s/\mu_i$ have $s\ne i$. Consequently no generating monomial of $\mathcal V$ equals $abc$. Linear independence of distinct Laurent monomials gives $abc\notin\mathcal V$. The same observation would also exclude $abc$ if products of two $\mathcal Q^+$ monomials were included, as in the independence check in the source appendix.

To bound accidental collisions, we clear denominators
before applying the Schwartz--Zippel bound. Put $D_\mu=\prod_{i=1}^L\mu_i$. Multiplying a difference of source exponents by $D_\mu$ yields a polynomial of degree at most $L+2$. Multiplying a difference of target exponents by $D_\mu^2$ yields one of degree at most $2L+4$, also when the target includes $abc$ or an independent variable $\vartheta$. These multiplications preserve zero tests on the sampling domain. For a nonzero such polynomial, the Schwartz--Zippel bound on independent variables sampled from $\F$ or $\F^*$ is at most $(2L+4)/(q-1)$. There are $O(L^2+Q)$ encodings, hence $O((L^2+Q)^2)$ possible pairs whose evaluations could collide. Conditioned on no unintended collision, formal equality tests have the same outcomes in both challenge worlds: an expression $v+\lambda abc$ with $v\in\mathcal V$ can vanish identically only if $\lambda=0$ and $v=0$, exactly as for $v+\lambda\vartheta$. Coupling the random encodings therefore gives identical generic transcripts outside the collision event. A union bound in the two experiments proves the claim.\qed
\end{proof}

\subsection{Static Security of the Large-Universe Modification}\label{app:pkc-large-proof}
The construction is exactly the large-universe modification in Appendix~\ref{app:pkc-repair}, with the attribute-indexed $H_1$ convention made explicit in Appendix~\ref{app:pkc-notation}. We use the static security model specified in
Appendix~\ref{app:pkc-notation}, including its explicit
corrupt-set nontriviality condition and the source's
distinct-identity and matching-index-set requirements.

\begin{theorem}\label{thm:fe-large}
Under augmented $L$-DBDH, the modified large-universe construction is statically IND-CPA secure in the random oracle model of Appendix~\ref{app:pkc-notation}, for polynomial sharing width and at most $L$ challenge rows. Multiple attributes per authority and repeated attribute labels are allowed. An adversary of advantage $\epsilon$ yields an augmented $L$-DBDH distinguisher with difference-of-probabilities advantage $\epsilon$.
\end{theorem}
\begin{proof}
The reduction receives $(\mathcal D_L^+,T)$ and publishes the group description. Upon receiving the static commitments, put $\boldsymbol\Delta=\boldsymbol v_0-\boldsymbol v_1$ and $n=|\mathcal I^\star|$. Equal challenge vectors have zero advantage. Otherwise normalize the corrupt rows to obtain $N=MP$ and $h=d-\rank(M_J)$, as in Appendix~\ref{app:pkc-proof}. Assign each challenge row a distinct instance slot. Throughout this proof the transmitted policy header is $M$, while payload calculations use $N$. Define
\[
 \nu_{i,j}=\begin{cases}N_{i,j},&j\le h,\\0,&j>h,\end{cases}
 \quad X_\theta=\{i:\rho(i)=\theta\},\quad
 O_t=X_{\mathsf T(t)}\setminus\{i:\delta(i)=t\}.
\]
The vectors $\nu_i$ are truncated rows, not a replacement access matrix.

\paragraph{Authority keys and hash values.}
For every honest authority occurring in the policy, choose independent uniform offsets and set implicitly
\[
 \alpha_\theta=\widehat\alpha_\theta+a\sum_{s\in X_\theta}\mu_s\nu_{s,1},\qquad
 y_{\theta,j}=\widehat y_{\theta,j}+a\sum_{s\in X_\theta}\mu_s\nu_{s,j}\quad(2\le j\le d).
\]
Their public encodings use the supplied $g_1^{a\mu_s}$. Authorities outside the policy are generated honestly; corrupt public keys are used as supplied, without their discrete logarithms.

Use the $H_1,H_2,H_3$ programming from Appendix~\ref{app:pkc-proof}, with $H_1$ indexed by attribute $t$ rather than authority. In particular,
\[
 h_{1,t,k}=(g_2^b)^{\Delta_k}g_2^{\eta_{1,t,k}},\quad
 h_{2,j,k}=g_2^{\eta_{2,j,k}},\quad
 h_{3,\uid,\boldsymbol u,j,k}=(g_2^b)^{\Delta_k\gamma_j}g_2^{\eta_{3,\uid,\boldsymbol u,j,k}}
\]
for a matching-index query. For an unauthorized query, choose $\boldsymbol d_{\uid}$ with first coordinate one and orthogonal to every queried or corrupt row of $N$; such a vector exists by unauthorizedness. Since the corrupt rows span $\{0\}^h\times\F^{d-h}$, this vector satisfies $d_{\uid,j}=0$ for $j>h$. Put $\gamma_j=d_{\uid,j}$ for $j\ge2$. For an authorized admissible query put $\boldsymbol d_{\uid}=\boldsymbol e_1$ and $\gamma_j=0$, noting that $\boldsymbol\Delta\cdot\boldsymbol u=0$. All offsets are independent uniform field elements. Nonmatching-index queries and other oracle inputs are treated as in the small-universe proof. Distinct committed identities avoid conflicting $H_3$ programming.

For every attribute $t$ of an honest policy authority and every $j\in[d]$, program
\begin{equation}\label{eq:lu-R}
 R(t\|j\|\mathcal I^\star)
 =g_2^{\xi_{t,j}}\prod_{s\in O_t}(g_2^{a\mu_s})^{\nu_{s,j}},\qquad \xi_{t,j}\sample\F.
\end{equation}
All other $R$ inputs receive independent known-exponent random elements. Only polynomially many queried or used inputs need be instantiated. Record each full oracle input and reuse its value. Conditional on the instance exponents, all authority offsets and oracle offsets independently mask fixed shifts. Thus public keys and oracle values have their joint real distribution; this includes attributes absent from the policy but owned by a policy authority.

\paragraph{Matching-index key responses.}
Fix a queried attribute $t$ of an honest policy authority $\theta$, and abbreviate
\[
 D=\boldsymbol\Delta\cdot\boldsymbol u,\quad
 A=\sum_k u_k\eta_{1,t,k},\quad
 B_j=\sum_k u_k(\eta_{2,j,k}+\eta_{3,\uid,\boldsymbol u,j,k}).
\]
Set initial auxiliary exponents $\tau_j^0=-bD d_{\uid,j}$ for all $j\in[d]$, so the auxiliary key encodings are computable as
$Z_t^{(j),0}=(g_1^b)^{-D d_{\uid,j}}$.
The base key augmented by $\prod_j R(t\|j\|\mathcal I^\star)^{\tau_j^0}$ has an unavailable quadratic term only of the form
$
 abD\sum_{s:\delta(s)=t}\mu_s\nu_s\boldsymbol d_{\uid}.
$
For an unauthorized query every such row is orthogonal to $\boldsymbol d_{\uid}$; truncation does not affect the inner product. For an authorized admissible query $D=0$. This term therefore vanishes in either case. Explicitly the initial main key component is
\begin{equation}\label{eq:lu-key}
\begin{aligned}
 K_{t,\boldsymbol u}^{0}={}&
 g_2^{\widehat\alpha_\theta A+\sum_{j=2}^d\widehat y_{\theta,j}B_j}
 \prod_{s\in X_\theta}(g_2^{a\mu_s})^{\nu_{s,1}A+\sum_{j=2}^d\nu_{s,j}B_j}\\
 &\cdot(g_2^b)^{D(\widehat\alpha_\theta+
 \sum_{j=2}^d\widehat y_{\theta,j}\gamma_j-
 \sum_{j=1}^d\xi_{t,j}d_{\uid,j})}.
\end{aligned}
\end{equation}
The cancellation removes rows belonging to other attributes of the same authority through $R$, while orthogonality removes all occurrences of the queried attribute. It does not require injective labels.

For each individual key-generation call independently choose $\widetilde\tau_j\sample\F$, multiply $K_{t,\boldsymbol u}^0$ by $\prod_j R(t\|j\|\mathcal I^\star)^{\widetilde\tau_j}$, and multiply $Z_t^{(j),0}$ by $g_1^{\widetilde\tau_j}$. The effective $\tau_j=\tau_j^0+\widetilde\tau_j$ are jointly independent uniform scalars, even conditional on the instance and all public values.

For a nonmatching-index query to a policy authority, the hashes used in its base key have known exponents. Its $\G_2$ authority encodings are computable from $g_2^{a\mu_s}$, so that key is computable directly with fresh uniform $\tau_j$. For an authority outside the policy, known secret exponents suffice, including on matching-index queries with programmed hashes. These cases cover every honest key response; corrupt keys are generated by the adversary from its own information.

\paragraph{Initial challenge ciphertext.}
Choose a uniform bit $\beta$ and implicitly set
\[
 B^0=-abc\boldsymbol e_1\boldsymbol\Delta^\top,\qquad
 \boldsymbol w^{(j,k),0}=\boldsymbol f^{(j),0}=-ac\boldsymbol e_j\ (2\le j\le h),
\]
with zero initial masks for $j>h$. For each honest row put $r_i^0=c/\mu_i$ and $C_{2,i}^0=g_1^{c/\mu_i}$. For $\theta=\rho(i)$, define the computable elements
\begin{equation}\label{eq:lu-V}
\begin{aligned}
 V_{i,1}&=(g_1^{c/\mu_i})^{\widehat\alpha_\theta}
 \prod_{s\in X_\theta\setminus\{i\}}
 (g_1^{ac\mu_s/\mu_i})^{\nu_{s,1}},\\
 V_{i,j}&=(g_1^{c/\mu_i})^{\widehat y_{\theta,j}}
 \prod_{s\in X_\theta\setminus\{i\}}
 (g_1^{ac\mu_s/\mu_i})^{\nu_{s,j}}\quad(2\le j\le d).
\end{aligned}
\end{equation}
They encode $\alpha_\theta r_i^0-acN_{i,1}$ and $y_{\theta,j}r_i^0-ac\nu_{i,j}$, respectively. Set
\begin{equation}\label{eq:lu-challenge}
\begin{aligned}
 C_{0,k}^0&=g_T^{v_{\beta,k}}T^{-\Delta_k},\\
 C_{1,i,k}^0&=e(V_{i,1},h_{1,\delta(i),k})
                  e(g_1^c,g_2^a)^{N_{i,1}\eta_{1,\delta(i),k}},\\
 C_{3,i,j,k}^0&=e(V_{i,j},h_{2,j,k}),\qquad C_{4,i,j}^0=V_{i,j},\\
 C_{5,i,j}^0&=(g_2^{c/\mu_i})^{\xi_{\delta(i),j}}
 \prod_{s\in O_{\delta(i)}}(g_2^{ac\mu_s/\mu_i})^{\nu_{s,j}}.
\end{aligned}
\end{equation}
Indeed, the diagonal $abcN_{i,1}\Delta_k$ in the unexpanded expression for $C_1$ cancels with $N_iB^0$, leaving the displayed $acN_{i,1}\eta_{1,\delta(i),k}$ correction. The initial zero sharings cancel the diagonal
$ac\nu_{i,j}$ in $C_3,C_4$. Every remaining cross term has $s\ne i$. In $C_5$, $i\notin O_{\delta(i)}$, and the augmented instance provides the previously missing first factor. The formulas also apply to $j>h$, where $\nu_{s,j}=0$ and $R=g_2^{\xi_{t,j}}$.

For corrupt rows choose known uniform $r_i^0$. All initial sharing terms vanish because $N_{i,j}=0$ for $j\le h$. Compute $C_1,C_2,C_3,C_4$ directly using the supplied corrupt public keys and public hashes; set $C_{5,i,j}^0=R(\delta(i)\|j\|\mathcal I^\star)^{r_i^0}$. This requires no corrupt secret exponent.

\paragraph{Restoring the joint distribution.}
Independently choose $\widetilde B\sample\F^{d\times n}$, every $\widetilde r_i\sample\F$, and every $\widetilde{\boldsymbol w}^{(j,k)},\widetilde{\boldsymbol f}^{(j)}\sample Z$. Apply the multipliers in~\eqref{eq:fe-rerand}, using the attribute-indexed $h_{1,\delta(i),k}$ in $D_{1,i,k}$, and additionally multiply each auxiliary component by
$
 D_{5,i,j}=R(\delta(i)\|j\|\mathcal I^\star)^{\widetilde r_i}.
$
All these operations use known scalars and public elements. Conditional on the instance and all preceding values, the effective encryption coins are jointly uniform on the modified scheme's product sampling space. In particular, each zero-sharing vector is shared across all rows for its fixed superscript, and its independent offset masks the fixed initial vector. Proposition~\ref{prop:fe-invariance} transfers the complete payload distribution from $N$ to the unchanged header $M$; the auxiliary components remain unchanged under the coupling.

\paragraph{Distinguishing.}
For $T=g_T^{abc}$, the resulting view is exactly the real static experiment for the modified construction. For $T=g_T^\vartheta$ with independent uniform $\vartheta$, every response other than $C_0$ is independent of $\vartheta$ and $\beta$. Conditional on these other values and all other simulator coins, $C_0$ has exponent vector
$\boldsymbol v_\beta+\widetilde{\boldsymbol z}-\vartheta\boldsymbol\Delta$.
The two bits induce the same affine-line distribution, since $\boldsymbol v_0-\boldsymbol v_1=\boldsymbol\Delta$. The adversary consequently guesses with probability $1/2$ in the random-target experiment. Outputting one precisely when its guess equals $\beta$ gives distinguishing advantage $\epsilon$.\qed
\end{proof}

\paragraph{Scope of the repaired result.}
Theorem~\ref{thm:fe-large} applies to the full zero-sharing
modification under augmented $L$-DBDH.
The generic-group analysis does not derive this assumption
from the original one or establish hardness in every
concrete group family.
The bound $L$ concerns challenge rows, not the size
of the attribute universe.

\section{Reconstruction of the Original Security Argument}\label{app:originalproof}
Theorem 5.1 of~\cite{DKW23} claims static security of the original construction under DBDH in the random oracle model. We reconstruct the simulation, including its basis-change step, to separate the computable parts of the argument from the distributional claim that will fail. The algebra below uses a consistent sign convention for the message mask.

\paragraph{Step 1: normalize corrupt rows.}
Given a DBDH instance $(g,g^a,g^b,g^c,T)$, the simulator receives the static commitments from Section~\ref{sec:static}. Put $J=\rho^{-1}(\mathcal C)$, $c_0=\rank(M_J)$, and $h=d-c_0$. The LSSS normalization lemma, with the proof correction from~\cite{DKWLWE21}, yields an equivalent matrix $N$ with $N_{i,k}=0$ for $i\in J$ and $k\le h$; its corrupt rows span the final $c_0$ coordinates. The ordinary share distribution for a fixed secret is preserved. The original argument then proceeds using $N$ on the premise that this preserves the adversary's view. Precisely this extension from shares to the full ciphertext is examined in Section~\ref{sec:rootcause}.

\paragraph{Step 2: generate honest public keys.}
For an honest authority $u=\rho(i)$ in the challenge policy, sample fresh uniform $\alpha'_u$ and $y'_{u,k}$. Implicitly choose
\begin{equation}\label{eq:originalembed}
 \alpha_u=\alpha'_u+abN_{i,1},\qquad
 y_{u,k}=\begin{cases}y'_{u,k}+aN_{i,k},&2\le k\le h,\\y'_{u,k},&h<k\le d.\end{cases}
\end{equation}
The public key is computable as $A_u=\E^{\alpha'_u}e(g^a,g^b)^{N_{i,1}}$, and $Y_{u,k}=g^{y'_{u,k}}(g^a)^{N_{i,k}}$ for $k\le h$, with $Y_{u,k}=g^{y'_{u,k}}$ otherwise. Authorities absent from the challenge are generated normally. Corrupt public keys are retained without extracting their logarithms. Each honest exponent is uniform because it contains an independent uniform offset.

\paragraph{Step 3: program identity hashes.}
For each identity with committed key set $U_{\uid}$, suppress the identity subscript locally and choose a vector $\boldsymbol t$ with $t_1=1$ orthogonal to the rows labeled by $\mathcal C\cup U_{\uid}$. Since the corrupt rows span the trailing subspace, $t_k=0$ for $k>h$. Take independent uniform $\eta_2,\ldots,\eta_d$ and set $H(\uid,k)=(g^b)^{t_k}g^{\eta_k}$ for $2\le k\le h$, and $H(\uid,k)=g^{\eta_k}$ for $k>h$. The values are independent uniform group elements. An identity without relevant key queries can be handled using an arbitrary vector orthogonal to the corrupt rows with first coordinate one. Hash values are recorded and reused consistently.

\paragraph{Step 4: answer key queries.}
For a queried honest authority outside the policy, the simulator knows its secret exponents and computes~\eqref{eq:originalkey} directly. For $u=\rho(i)$ in the policy, expanding the desired exponent gives an unknown $ab$ term with coefficient $N_{i,1}+\sum_{k=2}^hN_{i,k}t_k=N_i\boldsymbol t=0$. Consequently the key can be computed as
\begin{equation}\label{eq:originalsimkey}
 g^{\alpha'_u}\prod_{k=2}^h
 \left((g^b)^{t_ky'_{u,k}}(g^a)^{\eta_kN_{i,k}}g^{\eta_ky'_{u,k}}\right)
 \prod_{k=h+1}^d g^{\eta_ky'_{u,k}}.
\end{equation}

\paragraph{Step 5: prepare a challenge with special coins.}
The simulator implicitly sets $\boldsymbol v^0=-abc\boldsymbol e_1$, $x_k^0=-ac$ for $2\le k\le h$, and $x_k^0=0$ for $k>h$. For an honest row it sets $r_i^0=c$. The ciphertext components are then computable as
\begin{equation}\label{eq:originalbase}
 C_{1,i}^0=e(g^c,g)^{\alpha'_{\rho(i)}},\qquad
 C_{2,i}^0=g^c,\qquad C_{3,i,k}^0=(g^c)^{y'_{\rho(i),k}}.
\end{equation}
Indeed, the $-abcN_{i,1}$ share cancels the $abN_{i,1}$ term multiplied by $c$, and $-acN_{i,k}$ cancels the embedded coordinate term. For a corrupt row, all such shares vanish by normalization. The simulator chooses $r_i^0$ itself and computes $C_{1,i}^0=A_{\rho(i)}^{r_i^0}$, $C_{2,i}^0=g^{r_i^0}$, and $C_{3,i,k}^0=Y_{\rho(i),k}^{r_i^0}$ directly from the supplied public key.

Choose $\beta\sample\{0,1\}$ and put $C_0^0=m_\beta T^{-1}$. This is the encoding consistent with $v_1^0=-abc$ when $T=\E^{abc}$. The printed proof in~\cite{DKW23} uses $m_\beta T$ with this negative sharing exponent. Correcting that sign does not address the matrix-distribution issue studied here; the attack applies to the actual encryption algorithm regardless of this proof-level sign choice.

\paragraph{Step 6: re-randomize within the chosen representation.}
Sample independent offsets $\delta\boldsymbol v\sample\F^d$, $\delta x_k\sample\F$, and $\delta r_i\sample\F$. Update the initial payload by
\begin{equation}\label{eq:originalrerand}
\begin{aligned}
 C_0&=C_0^0\E^{\delta v_1},
 C_{1,i}=C_{1,i}^0\E^{N_i\delta\boldsymbol v}A_{\rho(i)}^{\delta r_i},\\
 C_{2,i}&=C_{2,i}^0g^{\delta r_i},
 C_{3,i,k}=C_{3,i,k}^0Y_{\rho(i),k}^{\delta r_i}g^{N_{i,k}\delta x_k}.
\end{aligned}
\end{equation}
Adding these independent offsets makes the effective sharing vector, coordinate scalars, and row coins independent and uniform. Thus, after the sign correction, a real DBDH target produces a properly distributed encryption payload for $N$. A random target makes $C_0$ independent of the message bit.
The reconstructed simulation produces a correctly distributed
payload for $N$, but does not establish the distribution required
for $M$ in the original scheme, as shown in
Section~\ref{sec:rootcause}.
Our repair resolves this issue by replacing coordinate-wise
masks with full zero sharings and proving the resulting
payload invariance.

\end{document}